\documentclass[11pt]{amsart}
\usepackage{amsmath,amssymb,amsthm,amsfonts}
\usepackage{mathtools}
\usepackage{pgfplots}
\pgfplotsset{compat=1.18}
\usepackage{tikz}
\usepackage{float}
\usepackage{enumitem}
\usepackage[T1]{fontenc}
\usepackage{xcolor}
\usepackage{hyperref}
\usepackage[backend=biber, style=numeric]{biblatex}
\newtheorem{theorem}{Theorem}[section]
\newtheorem{lemma}[theorem]{Lemma}
\newtheorem{proposition}[theorem]{Proposition}
\newtheorem{corollary}[theorem]{Corollary}

\newtheorem{claim}[theorem]{Claim}

\theoremstyle{definition}

\newtheorem{assumption}[theorem]{Assumption}
\newtheorem{remark}[theorem]{Remark}

\DeclareMathOperator{\Var}{Var}

\DeclareMathOperator{\GOE}{GOE}
\DeclareMathOperator{\Sp}{Sp}
\DeclareMathOperator{\GL}{GL}
\DeclareMathOperator{\SL}{SL}

\newcommand{\R}{\mathbb{R}}
\newcommand{\E}{\mathbb{E}}

\newcommand{\Id}{\mathrm{Id}}

\title{Quantitative Furstenberg Theory for Large Transfer Matrices}
\author{Reuben Drogin}

\begin{document}

\begin{abstract}

We consider the $2W\times 2W$ transfer matrices associated to the block Anderson
model with GOE potential blocks. We make the classical Lyapunov exponent
theory quantitative in two ways. First, we prove a quantitative limit theorem for the top
Lyapunov exponent. Second, we prove every gap between
Lyapunov exponents is at least $c/W$.
As a corollary, this implies 
the localization length of this $1$d block Anderson 
model is $CW^2$. 
The proof uses Furstenberg type formulas 
for the Lyapunov exponents, 
and Malliavin calculus style
arguments in the symplectic group 
to show the product of sufficiently many 
transfer matrices 
has a sufficiently smooth density. The main 
technical input for the latter 
is a least singular value 
estimate for a structured random matrix. 
\end{abstract}

\maketitle

\section{Introduction}

It is a classical fact that the asymptotic rates of expansion 
and contraction of random matrix products are governed by the Lyapunov exponents. 
Precisely, for any probability measure $\mu$ 
on $\SL_d(\mathbb{R})$ with $\int \log \|g\|\mu(dg)<\infty$, 
there exists deterministic
$L_1\geq\cdots\geq L_d$ such that 
if $(M_i)_{i\in \mathbb{N}}$ are independent samples from $\mu$, then
$$L_i = \lim_{n\to\infty}\frac{1}{n}\log \sigma_i(M_1...M_n).$$
Here $\sigma_i(M)$ denotes the $i$-th singular value of $M$, 
and the above limit holds almost surely. 
The existence of these limits follows from the
Furstenberg--Kesten theorem \cite{FurstenbergKesten1960}, or 
Kingman's subadditive ergodic theorem.

Given their existence, the fundamental question in most situations is 
to estimate the values of the $L_i$, either showing 
that $L_1$ is \textit{positive} or that the $L_i$ 
are \textit{simple}, i.e. $L_i-L_{i+1}>0$. 
In this random setting, when $\mu$ is sufficiently non-degenerate, 
the positivity of $L_1$ 
follows from Furstenberg's Theorem \cite{Furstenberg1963}, and 
the simplicity of the $L_i$ follows from the work of 
Goldsheid--Margulis \cite{Margulis87}. However, 
both of these landmark results are \textit{non-quantitative}, 
while many applications require or would benefit from a quantitative estimate. 
One goal of the present work is to present a method 
for making these results 
quantitative. 
For concreteness, we consider the transfer matrices 
of the block Anderson model.

\subsection{Main Results}
For any $E\in \R$, $W\in \mathbb{N}$ and $A\in \mathbb{R}^{W\times W}$ 
define the $2W\times 2W$ \textit{transfer matrix}, given by 
\begin{equation}
  \label{def:M}
  T_E(A):= \begin{pmatrix}
    A - E\,\Id_W & -\Id_W \\
    \Id_W & 0
  \end{pmatrix}. 
\end{equation}
A model case to have in mind is when $A$ has the normalized $\GOE$ distribution, i.e. 
$A\in \mathbb{R}_{\mathrm{Sym}}^{W\times W}$, with
$A_{xy} \sim \mathcal N(0, \frac{1+\delta_{xy}}{W})$, and $(A_{xy})_{1\leq x\leq y\leq W}$ independent.

Our first theorem makes Furstenberg's Theorem quantitative for this matrix.
\begin{theorem}\label{thm:top}
Fix $E\in \mathbb{R}$ and $K>0$. 
There exists $C:= C(E,K)>0$, 
such that if $W\geq 1$, 
$A\in \mathbb{R}^{W\times W}$ is GOE distributed or satisfies assumption \ref{assumption:top-LE} 
with constant $K$, and $L_1(E,W)$ is the 
largest Lyapunov exponent associated to the random matrix $T_E(A)$, given in \eqref{def:M}, 
then
\[
  \left|L_1(E,W)-\frac{1}{2}\log \lambda_E\right|\leq \frac{C}{W},
\]
where $\lambda_E$ is the top eigenvalue of
\begin{equation}
\label{def:M-E}
  M_E=
  \begin{pmatrix}
    1+E^2 & 1 & 2E \\
    1 & 0 & 0 \\
    -E & 0 & -1
  \end{pmatrix}.
\end{equation}
\end{theorem}

\medskip
\noindent\textit{Remark.}
Assumption \ref{assumption:top-LE}
is significantly weaker than 
GOE, and  
includes singular entry distributions 
such as centered Bernoulli random variables. 
\medskip

Our second theorem gives a quantitative version of
Goldsheid--Margulis's theorem on simplicity of the Lyapunov exponents
for the matrices in \eqref{def:M}. 

\begin{theorem}\label{thm:gaps}
Fix $E\in \mathbb{R}$. 
There exists $c := c(E)>0$ 
such that if $W\geq 1$, 
$A\in \mathbb{R}^{W\times W}$ is a normalized GOE matrix,
and $L_1(E,W)\geq \cdots \geq L_{2W}(E,W)$ are the 
Lyapunov exponents 
associated to the random matrix $T_E(A)$, given in \eqref{def:M}, then 
\[
  \min_{k\in [1,2W-1]}\left( L_k(E,W)-L_{k+1}(E,W)\right)\ge \frac{c}{W}.
\]
\end{theorem}

\begin{remark}
The proof can be applied to a broader class of transfer matrices. 
For instance
one may relax the GOE assumption and allow
non-identity, but still uniformly elliptic, hopping blocks. 
The choice of $\eqref{def:M}$ is made to
present the main ideas in the simplest setting.
\end{remark}

The above result is sharp in $W$ up to constants. 
Theorem \ref{thm:gaps} implies the Greens function 
of the $1$d Block Anderson model decays exponentially 
at the scale $W^2$, establishing a version
of the Fyodorov--Mirlin $\sqrt{N}$-conjecture 
\cite{FyodorovMirlin1991,Mirlin2000,Shapiro}. 

\begin{corollary}
\label{cor:RBM-Localization}
Let $W\in \mathbb{N}$,
and let $A_1, A_2,\ldots \in \mathbb{R}_{sym}^{W\times W}$ be independent 
normalized GOE matrices. 
If for each $n$
$H_{W,n}\in (\mathbb{R}^{W\times W})^{n\times n}$ 
denotes the random symmetric matrix 
$$
H_{W,n}:=
\begin{pmatrix}
  A_1 & \Id_W & 0 & \cdots & 0 \\
  \Id_W & A_2 & \Id_W & \ddots & \vdots \\
  0 & \Id_W & A_3 & \ddots & 0 \\
  \vdots & \ddots & \ddots & \ddots & \Id_W \\
  0 & \cdots & 0 & \Id_W & A_n
\end{pmatrix},
$$
then, for each $E\in \mathbb{R}$ we have,
$$\lim_{n\to \infty} \frac{1}{n}\log \|(H_{W,n}-E)^{-1}(1,n)\|_{op}\leq \frac{-c}{W},$$
almost surely. Here $(H_{W,n}-E)^{-1}(1,n)$ denotes the $(1,n)$ block 
in the matrix $(H_{W,n}-E)^{-1}$, and $c>0$ depends only on $E$.
\end{corollary}

\begin{remark}
An inspection of the proof gives that the constants 
are uniform for $E$ in a compact interval. 
Also, one can prove a high probability version of 
Corollary \ref{cor:RBM-Localization}
using the ideas in this work, 
although we do not pursue that here. 
\end{remark}

One ingredient to the proof of Theorem \ref{thm:gaps}, 
is the following. It lower bounds the gaps between
Lyapunov exponents by the smoothness of the law of the random matrix. 
To state it, we introduce the following notation. For any $M\in \mathfrak{sp}_{d}$
(the lie algebra of the $d$-dimensional Sympletic group) let 
$X_M$ be the left-invariant vector field generated by $M$, 
defined by 
\begin{equation}
\label{def:vector-field-derivative}
(X_M\phi)(g):= \frac{\partial}{\partial t}|_{t=0}\phi(g\exp(tM)).
\end{equation}
Further, we identify any $M\in \mathfrak{sp}_{d}$ as a 
matrix in $\mathbb{R}^{d\times d}$ and define
the Hilbert-Schmidt norm $\|M\|_{HS}$ in the standard way. 

\begin{theorem}
\label{thm:Fisher-gaps}
Let $m$ be any Haar measure on  $\Sp_{d}(\mathbb{R})$, 
and $\mu$ be a probability measure on $\Sp_{d}(\mathbb{R})$ 
with 
density $f = \frac{d\mu}{dm}$. Assume
$\mathbb{E}_{g\sim\mu}\log \|g\|<\infty.$ If
$L_1\geq \cdots \geq L_{d}$ are the Lyapunov exponents
of $\mu$, then
\[
  \min_{k\in [1,d-1]}(L_k-L_{k+1})
  \geq
  c\left(1+\sup_{\substack{M\in\mathfrak{sp}_{d},\\ \|M\|_{HS}=1}}\|X_Mf\|_{L^1(\Sp_{d}(\mathbb R))}\right)^{-2}
\]
where $c>0$ is an absolute constant, \(L^1(\Sp_{d}(\mathbb R))\)
means the \(L^1\) space with respect to \(m\), and we set
\(\|X_M f\|_{L^1(\Sp_{d}(\mathbb R))}=\infty\) if \(f\) is not weakly
differentiable. 
\end{theorem}

The theorem extends to $\GL_{d}(\mathbb{R})$ and 
$\SL_{d}(\mathbb{R})$ with straightforward modifications.

\subsection{Examples and Discussion}
Briefly we comment on using Theorem \ref{thm:Fisher-gaps}, 
and its application 
to Theorem \ref{thm:gaps}. 

When $\mu$ is smooth enough 
Theorem \ref{thm:Fisher-gaps} can be used immediately. 
For instance, consider the random 
matrix $Z\in GL_d(\mathbb{R})$ given by 
\begin{equation}
\label{eq:example-law} 
Z= M_0 + \delta B,
\end{equation}
for some $\delta>0$, deterministic $M_0\in \GL_{d}(\mathbb{R})$, 
and $B\in \mathbb{R}^{d\times d}$ Ginibre, i.e. having i.i.d. real 
Gaussian entries
with variance $1/d$. The RHS in 
Theorem \ref{thm:Fisher-gaps} can be computed 
via the density. Indeed, with
respect to Haar measure on 
$\GL_{d}(\mathbb{R})$, 
$Z$
has density proportional to 
\[
  f(Z)
  =
  |\det Z|^d
  \exp\left(-\frac{d}{2\delta^2}\|Z-M_0\|_{HS}^2\right),
\]
where the $|\det Z|^d$ comes from Haar measure, 
and the second term is the law of $\delta B$ on $\mathbb{R}^{d\times d}$. 
Using \eqref{def:vector-field-derivative}, 
a computation gives that
for any $M\in \mathfrak{gl}_d$,
$$\|X_Mf\|_{L^1(\GL_{d}(\mathbb{R}))}\leq C\delta^{-1}\sqrt{d}(\delta+\|M_0\|_{op})\|M\|_{HS}.$$
Hence, the $GL_d(\mathbb{R})$ version of Theorem \ref{thm:Fisher-gaps} implies the Lyapunov exponents of 
\eqref{eq:example-law} are seperated on the scale $\delta^2 d^{-1}$, 
when $c\leq \|M_0\|_{op}\leq C$, and $\delta\in [0,1]$. 

In some similar Gaussian cases 
more precise results are known, see \cite{Newman1986, Forrester1, Forrester2,ForresterZhang2020}. 
However, a similar computation as above proves the same result 
for \textit{non-gaussian} distributions with a sufficenlty smooth density, 
which seems to be new. 
We note that the $\delta^2$ scaling here 
is related to the main result of Bednarski, Dewitt, and Quas
\cite{QuasDewitt}. Their result
applies to the more general non-stationary product setting, 
where $M_0$ may change between terms in a random product, 
but ours is effective in dimension. 

The other extreme is when $\mu$ is a 
discrete measure, 
say supported on finitely many elements 
of $\Sp_{d}(\mathbb{R})$. In this case   
Theorem \ref{thm:Fisher-gaps} 
cannot be applied at all as any convolution power of $\mu$ will be discrete. 
For such measures, a useful replacement for the right-hand side
of Theorem \ref{thm:Fisher-gaps} is the relative 
entropy quantity
\[
  \mathbb{E}_{g,b}\left[
    \operatorname{KL}\left(\nu_b \mid g_{\ast}^{-1}\nu_{gb}\right).
  \right],
\]
proven in the work of Lessa \cite{Lessa2021}. Here $\operatorname{KL}$ 
is the Kullback-Liebler divergence, $b$ is a random flag, $(\nu_b)$ is a 
family of measures indexed by flags and $g_{\ast}^{-1}$ is the natural pullback action of $g$. 
See Proposition \eqref{prp:KL-gap-formula} and 
surrounding arguments for more details, or consult the work of Lessa \cite{Lessa2021}. 
Lower bounds for this quantity are closely
related to spectral gaps and mixing for actions 
generated by finitely many matrices
on projective or flag spaces.  
This is a difficult problem in the theory
of random walks and expansion in matrix groups; for instance see
\cite{BourgainGamburd2008,BreuillardGamburd2010,
LindenstraussVarju2016,Kittle2025,Kogler2025,
BourgainGamburd2012}.

Many cases of interest, such as our \eqref{def:M}, 
lie between these two, having smoothness in some directions but not every direction. 
Indeed, 
the law 
of \eqref{def:M} is not smooth
in $\Sp_{2W}$; it is supported on a submanifold 
of dimension $W(W+1)/2$. In such cases, one may hope that 
the law of a product of samples, given by $\mu^{\ast N_0}$ for some $N_0$, may have a 
weakly differentiable density. This can be seen 
as a hypoelliptic problem, and is the approach we take 
to proving Theorem \ref{thm:gaps} using Theorem \ref{thm:Fisher-gaps}. 
In our case, estimating the derivatives 
of $\mu^{\ast N_0}$, \textit{quantitatively} in $W$, 
presents a real challenge and is another contribution of the current work. 
This is explained below, along with the sketch of the proofs 
of Theorem \ref{thm:top} and Theorem \ref{thm:gaps}.

\subsection{Proof Sketch} 

Theorem \ref{thm:top} uses Furstenberg's formula together with an estimate
on the stationary measure.  If \(v=(x,y)^T\in \R^{2W}\) is a unit
vector, then with high probability
\[
\begin{aligned}
  \|T_E(A)v\|^2
  &=\|(A-E\Id_W)x-y\|^2+\|x\|^2  \\
  &= \|(A-E\Id_W)x\|^2 + \|y\|^2 - 2\langle (A-E\Id_W)x, y\rangle + \|x\|^2\\
  & \approx (2+E^2)\|x\|^2+\|y\|^2 +2E\langle x,y\rangle,
\end{aligned}
\]
where $\approx$ means up to random mean $0$ errors 
of typical size $W^{-1/2}$.  
In the last line, we used that \(\|Ax\|^2\approx \|x\|^2\) and 
\(\langle Ax,y\rangle\approx 0\). 

Furstenberg's formula for the top Lyapunov exponent (Proposition \ref{prp:Furstenberg-formula-top}) 
then reduces the problem to 
understanding the distribution of the three quantities
\[
  \|x\|^2,\ \|y\|^2,\ \langle x,y\rangle
\]
when $v = (x,\ y)^T$ is a unit vector whose 
direction is drawn from the 
stationary distribution of $T_E(A)$ on $\mathbb{P}(\mathbb{R}^{2W})$. 
The main observation we use is that by stationarity we have 
$$v \stackrel{d}{=} \frac{T_E(A)v}{\|T_E(A)v\|},$$
and thus a computation, see Lemma \ref{lem:recurrence}, then gives that
if 
$$X:=\left(\|x\|^2,\ \|y\|^2,\ \langle x,y\rangle\right)^T,$$
then
\[
\begin{gathered}
  X\stackrel{d}{=}
  \frac{M_E X+Y}{(M_E X+Y)_1+(M_E X+Y)_2},\\
\end{gathered}
\]
where $Y$ is a centered random variable with typical size $W^{-1/2}$, 
and \(M_E\) is given in \eqref{def:M-E}. $M_E$ has a unique expanding 
direction, and so one can argue, see Proposition \ref{prop:stat}, 
that $X$ must be concentrated near that direction.

Theorem \ref{thm:gaps} uses Theorem \ref{thm:Fisher-gaps} (which also 
relies on a Furstenberg-type formula) to reduce lower bounding the 
gaps to estimating the \textit{smoothness} 
of the law of products of independent samples of $T_E(A)$.
To do this
we use ideas from the theory of hypoellipticity and Malliavin calculus.  

To illustrate the basic mechanism, consider 
the random variable 
$f(Z)$ where $f:\mathbb{R}^d\to \mathbb{R}^N$ is the linear function
\[
  f(z)=\sum_{i=1}^d z_i v_i,
\]
for some vectors
$v_1,...,v_d\in \mathbb{R}^{N}$, and \(Z = (z_1,...,z_d)\) is a standard Gaussian vector in
\(\mathbb R^d\). 
If the vectors \(v_i\) do not span \(\mathbb R^N\), then the random variable 
\(f(Z)\) is
supported on a proper subspace and has no density on \(\mathbb R^N\).
If they do span $\mathbb{R}^N$, then \(f(Z)\) is a non-degenerate Gaussian with covariance
\[
  \Gamma=Df\,(Df)^T=\sum_{i=1}^d v_i v_i^T,
\]
and density given by  
\[
  \frac{1}{(2\pi)^{N/2}\det(\Gamma)^{1/2}}
  \exp\left(-\frac12 \langle \Gamma^{-1}x,x\rangle\right).
\]
Hence, the smoothness of the density 
is governed by the smallest eigenvalue of $\Gamma$. 

In our problem, we're considering the random variable 
$F_E(a_1,...,a_{1000})$ where   
 $F_E:(\mathbb{R}^{W\times W}_{sym})^{1000}\to \Sp_{2W}(\mathbb{R})$
is the function
\[
  F_E(a_1,\ldots,a_{1000})=T_E(a_1)\cdots T_E(a_{1000})
\]
and $a_1,...,a_{1000}$ are independent Gaussian random matrices, 
say GOE distributed.  
Despite this different setting, the main 
technical estimate is precisely an estimate on the 
smallest eigenvalue of the so called 
\textit{Malliavin covariance} 
$$\Gamma(a):= DF_E(a)DF_E(a)^{\ast}.$$
This random matrix is non linear in the $a_i$, 
and highly structured due to the non-isotropic 
form of the random matrices $T_E(a_i)$. 
Nevertheless, it is possible to show 
that its smallest eigenvalue is bounded below by a positive constant with high probability 
by an epsilon net argument, following 
ideas of \cite{RudelsonVershynin2009,TaoLeastSingularValue2010}. 

We note that a further difficulty of the non-linear 
structure of $F_E$ is that there is no accessible
form of the density. Hence,
even having controlled $\Gamma(a)$, 
it is not clear how to estimate the derivatives 
of the density which appear in Theorem \ref{thm:Fisher-gaps}. 
For this we use an 
idea from Malliavin calculus (see Section~3 in \cite{Hairer2011}), 
which takes the form of Lemma \ref{lem:abstract-gaussian-lifting}. 
The idea is to estimate
the smoothness of  
$\frac{d\mu^{\ast 1000}}{dm}(g)$
under the flow generated by the left invariant vector field $X_M$ 
by finding a vector field $\tilde X_M:\mathcal H_W \to \mathcal H_W$ 
which pushes forward under $F_E$ to the same flow, i.e. 
$\tilde X_M$ satisfies
$$DF_E(a) \tilde{X}_M(a) = X_M(F_E(a)).$$
Indeed, along any solution $a(t)\in \mathcal H_W$ to this ODE 
one has 
$$F_E(a(t)) = F_E(a(0))\exp(tM).$$
Thus the problem reduces to estimating the 
smoothness of the Gaussian density on $\mathcal H_W$ 
under the vector field $\tilde X_M$. Malliavin calculus suggests that we take 
\begin{equation}
\label{eq:deformation}
\tilde X_M(a):= DF_E(a)^{\ast}\Gamma(a)^{-1}X_M(F_E(a))
\end{equation}
The size and smoothness of this vector 
field are controlled by $\|\Gamma(a)^{-1}\|$, 
and so our estimate on the least singular value of 
$\Gamma(a)$ completes the proof.

\subsection{Related Work}
We mention some related works,
sorted into three broad groups. 

\textit{Lyapunov Exponents of Large transfer matrices.}
For several classical highly isotropic ensembles Newman
\cite{Newman1986} computed exact
Lyapunov exponent formulae, and 
in particular noted such models tend to have Lyapunov exponent gaps of order \(1/d\) in dimension \(d\).
Recently, Shapiro noted \cite{Shapiro} 
that this implies a version of the Fyodorov-Mirlin $\sqrt{N}$ 
conjecture for a special RBM model when $E=0$. 
Isopi--Newman also proved the empirical distribution 
of the Lyapunov exponents
of isotropic random matrix products 
follows the 
triangular law in the
large-dimensional limit \cite{IsopiNewman1992}. See also \cite{HaninJiang2025}.  

Specific to Anderson and Random band models, we 
mention the recent works 
\cite{Schenker,PeledSchenkerShamisSodin2019, SchenkerPeledCipolloni, SmartChen,Goldstein2022RBM, DroginRBMLocalization2025},
which proved localization for $1d$ random band matrices using 
a different approach based on a factorization of the Greens function.
Perhaps the closest result in this class to the present work 
is the analysis of Schulz-Baldes \cite{SchulzBaldes2004}
for the Anderson model on the strip.  
That work proves the smallest positive Lyapunov exponents 
of the associated transfer matrix scales like $\lambda^2/W$
in the fixed $W$, and small $\lambda$ regime. Here $\lambda$ is the disorder strength, 
and $W$ is the width of the strip. The proof 
also gives a quantitative analysis of the action of transfer matrices 
on the associated stationary measures. 
Quantitative estimates on localization lengths and transfer matrices in
related models also appear in
\cite{ShubinWolff1998,SchlagShubinWolff2002,JitomirskayaSchulzBaldesStolz2003,Bourgain2012,BinderGoldsteinVoda2015}.

\textit{Furstenberg formulae, entropy, and Lyapunov estimates.}
The qualitative positivity and simplicity theory for random matrix products
starts with the landmark Furstenberg's theorem and the Goldsheid--Margulis criterion
\cite{Furstenberg1963,Margulis87}. 
These results were used to give some of the earliest proofs of Anderson localization 
\cite{GoldsheidMolchanovPastur1977,Lacroix1984,KleinLacroixSpeis1990}
in one-dimensional models. 
One key feature of both the works of Furstenberg and Goldsheid--Margulis 
is the use of stationary measures and their relation to the Lyapunov exponents.  
This relation is captured in Furstenberg-type formulae, which express
Lyapunov exponents in terms of stationary measures on projective spaces or
flags. For example, see \cite{FurstenbergKifer1983,BougerolLacroix1985}.
Entropic versions of these formulae go back to Ledrappier's work
\cite{Ledrappier-1982} and also arise in Furstenberg's original paper \cite{Furstenberg1963}. 
We also mention the work of Lessa 
on exact dimensionality \cite{Lessa2021}, 
as it is 
closely related to the relative-entropy approach behind
Theorem \ref{thm:Fisher-gaps}. 

Recent work has extended these ideas to the non-stationary 
setting and to give quantitative estimates.  Gorodetski--Kleptsyn prove a
non-stationary Furstenberg theorem \cite{GorodetskiKleptsyn2026}, while
Bednarski--DeWitt--Quas obtain effective quantitative singular-value gap estimates for
non-stationary 
products with small independent noise \cite{QuasDewitt}.  

\textit{Hypoellipticity, Malliavin Calculus, and Lyapunov Exponents.}

Recently, a growing body of work on stochastically forced dynamical systems, 
especially fluid equations, has combined ideas from hypoellipticity 
\cite{Hormander1967} and Malliavin calculus \cite{Hairer2011} 
with Furstenberg-type and entropy formulas for Lyapunov exponents. 
Notably, Bedrossian--Blumenthal--Punshon-Smith 
\cite{BedrossianBlumenthalPunshonSmith2022} used 
these ideas to give a quantitative lower bound on the top Lyapunov exponent
 for a class of stochastically forced 
differential equations. Related developments for randomly forced fluid 
equations and other stochastic evolution equations appear in 
\cite{HairerPunshonSmithRosatiYi,HairerStacy2026,CoopermanRowan2025}.
We also mention the recent work of Pillai--Smith \cite{PillaiSmith2026Kac}, 
which uses ideas in Malliavin calculus on $SO(n)$ to show Kac's walk 
mixes in $n^2\log(n)$ steps.

\subsection{Acknowledgements}
The author thanks
Sebastian Hurtado  
for telling him about the work of Pablo Lessa, 
Felipé Hernandez for suggesting the use of strong convergence 
in the appendix, and Jonathan DeWitt and Anthony Quas for many helpful conversations 
relating to their work and the ideas in section 2. 
The author also thanks Charlie Smart for encouragement 
throughout this project and Adam Black 
for comments on an earlier draft.  

The mathematical content of this paper was developed by the author. 
Large language model tools aided in the preparation 
of the manuscript, including for literature search and proofreading.

\subsection{Conventions and Use of Constants}
Throughout, $C,c>0$ denote constants whose values may change from
line to line. Unless otherwise stated, these constants may depend on
the fixed parameters $E$ and $K$, but are independent of $W$.
Any additional dependence is indicated explicitly.

For quantities in a normed space, we write $X=O(r)$ if
$\|X\|\leq Cr$, with the same convention on $C$. For instance,
$X=Y+O(r)$ means $\|X-Y\|\leq Cr$.

We use $m$ to denote a Haar measure on $\Sp_{d}(\mathbb{R})$. 
Its scaling is fixed throughout the paper, and does not matter for us 
since $\|X_M\frac{d\mu}{dm}\|_{L^1(dm)}$ is constant under rescaling $m$. 

\section{Formulas for Lyapunov Exponents}\label{sec:furstenberg-formula}
Both Theorem \ref{thm:top} and \ref{thm:gaps} use a version of 
Furstenberg's formula, which relates the Lyapunov exponents of 
a random matrix to its stationary measures. 
For the top Lyapunov exponent, this formula dates back to 
Furstenberg's seminal paper \cite{Furstenberg1963}.

\begin{proposition}[Furstenberg's Formula]
\label{prp:Furstenberg-formula-top}
Let $\mu$ be a probability measure on $GL_{d}(\mathbb{R})$
whose support does not preserve 
any finite union of proper subspaces of $\R^d$
and for which $\mathbb{E}_{g\sim \mu}\log^{+} ||g||<\infty$. 
Further, let $\lambda_1$ denote the top Lyapunov exponent 
of $\mu$.  Then there exists a $\mu$-stationary measure $\nu$ on $\mathbb{P}(\R^d)$, 
and for any such measure we have
\[
  \lambda_1
  =
  \mathbb{E}_{g}\int_{\mathbb{P}(\R^d)}
  \log \frac{\|gv\|}{\|v\|}\,\nu(dv). 
\]
Here $\mathbb{E}_{g}$ 
denotes the integral over $GL_{d}(\mathbb{R})$ with respect to the measure $\mu$. 
\end{proposition}

\begin{remark}
In the above and what follows we commit a small abuse of notation 
by identifying vectors and subspaces with their equivalence classes in 
projective spaces. For instance in the above formula we have used $v$  
to denote both a vector in $\mathbb{R}^d$ and its equivalence class in $\mathbb{P}(\mathbb{R}^{d}).$ 
\end{remark}

Recall that we say 
a measure $\nu$ on $\mathbb{P}(\mathbb{R}^d)$ 
is \textit{$\mu$-stationary} if
$\nu = \mathbb{E}_{g}g_{\ast}\nu$, i.e. for any open set 
$S\subset \mathbb{P}(\R^d)$ 
\[
  \nu(S)=\mathbb{E}_{g}g_{\ast}\nu(S) = \int_{GL_d(\R)} \nu(g^{-1}S)\,\mu(dg).
\]
This is equivalent to $\nu$ being a stationary measure of the Markov process 
$v\to \frac{gv}{\|gv\|}$ on projective space. 
For a proof of Proposition \ref{prp:Furstenberg-formula-top}
we refer the reader to Proposition 7.2 in \cite{BougerolLacroix1985}.  

For Theorem \ref{thm:gaps} we need a variant of the above 
for the \textit{gaps between Lyapunov exponents}. This 
comes in the form of Theorem \ref{thm:Fisher-gaps}.
The proof combines the work of Lessa \cite{Lessa2021} 
with a one-dimensional non-concentration estimate.
In particular, we use Proposition \ref{prp:KL-gap-formula} below, 
which is a version of Theorem 1 in \cite{Lessa2021} 
for the Symplectic group. The proof is essentially 
the same as in \cite{Lessa2021}, but we include 
it for completeness.

To state it we introduce the following flag spaces
and the corresponding stationary measures. 
Let $\omega$ be the standard symplectic 
form on $\mathbb{R}^{2d}$, and define  
$\mathcal{F}_{0}:=\{0\},$
and $\mathcal F_k$ by 
\[
\mathcal{F}_{k}:= \{V\leq \mathbb{R}^{2d}: \dim(V) = k,\ \omega|_{V}=0\},
\] 
for \(1\leq k\leq d\).
Further, we let $\mathcal F_{d+1}$ denote the so-called 
coisotropic Grassmannian, given by
\[
\mathcal{F}_{d+1}:= \{V\leq \mathbb{R}^{2d}: \dim(V)=d+1,\ V^{\perp}\subset V\}.
\]
Here $\perp$ denotes symplectic orthogonal. Equivalently,
\[
\mathcal{F}_{d+1}=\{U^{\perp}:U\in\mathcal{F}_{d-1}\}. 
\] 

To estimate the gap between the lyapunov exponents $\lambda_k$ and $\lambda_{k+1}$, 
we need to consider the action of $\mu$ on \textit{pairs} and \textit{triples}
of subspaces. Hence, for any $1\leq k\leq d$, we define
the flag spaces
\[
\mathcal F_{(k-1,k+1)}
:=
\{(U,V)\in\mathcal F_{k-1}\times \mathcal F_{k+1}\ :
U\subset V,\ \text{and if }k=d,\ V=U^\perp\},
\]
and
\[
\mathcal F_{(k-1,k,k+1)}
:=
\{(U,W,V)\in \mathcal F_{k-1}\times \mathcal F_{k}\times \mathcal F_{k+1}\ :
U\subset W\subset V,\ \text{ and if } k=d,\ V = U^{\perp}\}.
\]
All spaces above are endowed with the natural topology and probability measure
induced by normalized Haar measure on the Symplectic rotation group
$\Sp_{2d}(\mathbb{R})\cap O(2d)$ through its action on subspaces.
Under the assumption that $\mu$ has a density on $\Sp_{2d}(\mathbb{R})$, 
it is standard that there is a unique $\mu$-stationary measure on each of the flag spaces above
which is absolutely continuous with respect to the natural measure on each such space.

The following gives an expression 
for the gaps between Lyapunov exponents. 
Roughly speaking, it says the gap \(\lambda_k-\lambda_{k+1}\) measures the typical 
distance between 
the family of conditional measures $(\nu_{V/U})_{(U,V)\in \mathcal F_{(k-1,k+1)}}$
and their pushforward $(g_{\ast}\nu_{V/U})_{(U,V)\in \mathcal F_{(k-1,k+1)}}$
under a random $g$. 

\begin{proposition}
\label{prp:KL-gap-formula}
Let $\mu$ be a probability measure on $\Sp_{2d}(\mathbb{R})$
that is 
absolutely continuous with respect to Haar measure
and for which $\mathbb{E}_g \log \|g\|< \infty$. Further, 
for any $k\in[1,d]$, let \(\nu_{(k-1,k,k+1)}\) be the unique \(\mu\)-stationary measure on
\(\mathcal F_{(k-1,k,k+1)}\), \(\nu_{(k-1,k+1)}\) be the unique
$\mu$-stationary measure on $\mathcal F_{(k-1,k+1)}$, and 
for \(\nu_{(k-1,k+1)}\)-almost every \((U,V)\), let
\(\nu_{V/U}\) be the conditional law of the intermediate subspace \(W\)
under \(\nu_{(k-1,k,k+1)}\), identifying \(W\) with the line
\(W/U\in\mathbb P(V/U)\).
Then for any $k\in [1,d]$ we have
\[
\lambda_k-\lambda_{k+1}
=
\int_{\mathcal F_{(k-1,k+1)}}
\mathbb E_g KL\!\left(
g_*\nu_{V/U}\mid \nu_{gV/gU}
\right)
\,\nu_{(k-1,k+1)}(d(U,V)).
\]
Here \(g_*\nu_{V/U}\) denotes the pushforward under
the natural mapping of $g$ from
\(\mathbb P(V/U)\to\mathbb P(gV/gU)\), and
\[
  KL(\eta\mid\xi)
  :=
  \int \log\!\left(\frac{d\eta}{d\xi}\right)\,d\eta, 
\]
is the $KL$ divergence between two measures. 
\end{proposition}

We have viewed the measures $(\nu_{V/U})_{(U,V)\in \mathcal F_{(k-1,k+1)}}$
as a family of measures on the projective spaces $(\mathbb{P}(V/U))_{(U,V)\in \mathcal F_{(k-1,k+1)}}$
which disintegrate $\nu_{(k-1,k,k+1)}$ with respect to $\nu_{(k-1,k+1)}$, 
i.e. for any  sufficiently regular $\Phi:\mathcal F_{(k-1,k,k+1)}\to \mathbb{R},$
we have 
\begin{equation}
\label{eq:distintegration}
\begin{aligned}
\int_{\mathcal F_{(k-1,k,k+1)}} & \Phi(U,W,V) \nu_{(k-1,k,k+1)}(d(U,W,V))\\
& = \int_{\mathcal F_{(k-1,k+1)}} \left(\int_{\mathbb{P}(V/U)} \Phi(U,W,V) \nu_{V/U}(dW)\right) \nu_{(k-1,k+1)}(d(U,V)).
\end{aligned}
\end{equation}
Again, in the inner integral on the RHS, we are committing the small abuse of notation 
of identifying the subspace $W$ with the corresponding point in $\mathbb{P}(V/U)$.

\begin{proof}[Proof of Proposition \ref{prp:KL-gap-formula}]
We wil use the following corollary of 
Proposition \ref{prp:Furstenberg-formula-top} 
applied to wedge products. 
For a proof we refer the reader to Proposition 3.4 in \cite{BougerolLacroix1985}. 
\begin{theorem}
\label{thm:furstenberg-formula-sums}
Let $\mu$ be a probability measure on $\Sp_{2d}(\mathbb{R})$
which has a density with respect to Haar measure, 
and for which $\mathbb{E}_{g}\log \|g\|< \infty$. 
If $\lambda_1\geq \lambda_2\geq\dots \geq \lambda_{2d}$ 
are the Lyapunov exponents of $\mu$ and for each $k\in [1,d+1]$
$\nu_k$ denotes the unique $\mu$-stationary measure, 
then 
$$\lambda_1+\cdots+\lambda_k = \mathbb{E}_{g}\int_{\mathcal{F}_{k}}\log |\det_{V} (g)| \nu_k(dV),$$
where $\det_{V}(g)$ denotes the determinant of the linear mapping $g:V\to gV$, 
with respect to the Euclidean volume on $\mathbb{R}^{2d}$.  
\end{theorem}

\begin{remark}
The integrand can be written in terms of wedge products as 
$$|\det_V(g)| = \frac{\|(\wedge ^k g)(v_1\wedge\dots \wedge v_k)\|}{\|v_1\wedge\dots \wedge v_k\|}= \frac{\|gv_1\wedge\dots \wedge gv_k\|}{\|v_1\wedge\dots \wedge v_k\|},$$
where $v_1,...,v_k$ is any basis for the subspace $V$. 
\end{remark}

Note that if $k=d+1$, this agrees with the symplectic identity
$\lambda_1+\cdots+\lambda_{d+1}=\lambda_1+\cdots+\lambda_{d-1}$, since
$\det_{U^\perp}(g)=\det_U(g)$ for $g\in\Sp_{2d}(\R)$.

First, we write the gap $\lambda_k-\lambda_{k+1}$ 
in terms of integrals over stationary measure 
(closely following the argument in Lessa \cite{Lessa2021}): 
\begin{align*}
&\lambda_{k}-\lambda_{k+1}\\
&=
2(\lambda_1+\cdots+\lambda_{k})
-(\lambda_1+\cdots+\lambda_{k-1})
-(\lambda_1+\cdots+\lambda_{k+1})\\
&=
\mathbb{E}_{g}\int_{\mathcal{F}_{(k-1,k,k+1)}}
\log \left|\frac{\det_{V_{k}}(g)^2}
{\det_{V_{k-1}}(g)\det_{V_{k+1}}(g)}\right|\nu_{(k-1,k,k+1)}(d(V_{k-1}, V_k, V_{k+1})).
\end{align*} 
To pass to the second line we
used that the marginals of $\nu_{(k-1,k,k+1)}$ are $\nu_{k-1}, \nu_{k}, \nu_{k+1}$.

Second, note that the integrand can be written in terms 
of the Jacobian of 
$g$ on the projective space $\mathbb{P}(V_{k+1}/V_{k-1})$. 
Indeed, if $g\in \Sp_{2d}(\mathbb{R})$, $V_{k-1}\subset V_k\subset V_{k+1}$ as above, 
and $m_{V/U}$ denotes the rotation invariant measure on 
the projective space $\mathbb{P}(V/U)$, then 
$$\frac{d(g_{\ast}m_{V_{k+1}/V_{k-1}})}{dm_{gV_{k+1}/gV_{k-1}}}(gV_{k}) = \left|\frac{\det_{V_k}(g)^2}{\det_{V_{k-1}}(g)\det_{V_{k+1}}(g)}\right|.$$
We have abused notation here by writing $gV_{k}$ for the 
element $gV_k/gV_{k-1}\in \mathbb{P}(gV_{k+1}/gV_{k-1}).$

Third, and finally, we note that 
because $\nu_{(k-1,k,k+1)}$ is $\mu$-stationary,
we have $(U,W,V) \stackrel{d}{=}(gU,gW,gV)$, 
when $g$ is sampled from $\mu$ and $(U,W,V)$ is sampled 
from $\nu_{(k-1,k,k+1)}$ independently. Thus
for any integrable function $\Phi:\mathcal{F}_{(k-1,k,k+1)}\to \mathbb{R}$, 
we have 
\begin{align*}
\mathbb{E}_g\int_{\mathcal F_{(k-1,k,k+1)}}&  \Phi(gU,gW,gV) \nu_{(k-1,k,k+1)}(d(U,W,V))\\
& = \int_{\mathcal F_{(k-1,k,k+1)}} \Phi(U,W,V) \nu_{(k-1,k,k+1)}(d(U,W,V)).
\end{align*}
It is standard that $\nu_{V/U}$ is absolutely 
continuous with respect to $m_{V/U}$ when $\mu$ is absolutely continuous with respect 
to Haar measure on $\Sp_{2d}(\mathbb{R})$, 
so we may apply this to the $\log$ density, 
$\rho_{V/U}(W) = \log \frac{d\nu_{V/U}}{dm_{V/U}}(W),$ 
to get
\begin{equation}
\label{eq:stationarity-identity}
0= 
\mathbb{E}_{g}\int_{\mathcal{F}_{(k-1,k,k+1)}}\left(
\rho_{gV_{k+1}/gV_{k-1}}(gV_k)
- \rho_{V_{k+1}/V_{k-1}}(V_k)
\right)
\nu_{(k-1,k,k+1)}(d(V_{k-1}, V_k, V_{k+1})).
\end{equation}
Combining the above displays gives
\begin{equation}
\label{eq:KL-gap-formula}
\begin{aligned}
&\lambda_k-\lambda_{k+1}\\
&=
\mathbb{E}_g\int_{\mathcal F_{(k-1,k,k+1)}}
\log\left(
\frac{d(g_*m_{V_{k+1}/V_{k-1}})}
{dm_{gV_{k+1}/gV_{k-1}}}(gV_k)
\right)
\nu_{(k-1,k,k+1)}(d(V_{k-1},V_k,V_{k+1}))\\
&=
\mathbb{E}_g\int_{\mathcal F_{(k-1,k,k+1)}}
\left[
\log\left(
\frac{d(g_*m_{V_{k+1}/V_{k-1}})}
{dm_{gV_{k+1}/gV_{k-1}}}(gV_k)
\right)
+\rho_{V_{k+1}/V_{k-1}}(V_k)
-\rho_{gV_{k+1}/gV_{k-1}}(gV_k)
\right]\\
&\qquad\qquad\qquad\qquad\qquad\qquad\qquad \qquad\qquad 
\qquad\qquad \qquad\qquad
\qquad\qquad  
\nu_{(k-1,k,k+1)}(d(V_{k-1},V_k,V_{k+1}))\\
&=
\mathbb{E}_{g}\int_{\mathcal{F}_{(k-1,k,k+1)}}
\log\left( \frac{d(g_*\nu_{V_{k+1}/V_{k-1}})}
{d\nu_{gV_{k+1}/gV_{k-1}}}(gV_{k})\right)
\nu_{(k-1,k,k+1)}(d(V_{k-1}, V_k, V_{k+1}))\\
&=
\mathbb{E}_{g}\int_{\mathcal{F}_{(k-1,k+1)}}
\int_{\mathbb{P}(V_{k+1}/V_{k-1})}
\log\left( \frac{d(g_*\nu_{V_{k+1}/V_{k-1}})}
{d\nu_{gV_{k+1}/gV_{k-1}}}(gV_{k})\right)
\nu_{V_{k+1}/V_{k-1}}(dV_k)
\nu_{(k-1,k+1)}(d(V_{k-1}, V_{k+1}))\\
&=
\int_{\mathcal{F}_{(k-1,k+1)}}
\mathbb{E}_{g}KL\left(g_*\nu_{V_{k+1}/V_{k-1}}\mid\nu_{gV_{k+1}/gV_{k-1}}\right)
\nu_{(k-1,k+1)}(d(V_{k-1}, V_{k+1})).
\end{aligned}
\end{equation}
For the second equality we used \eqref{eq:stationarity-identity}. 
For the third equality we combined the logarithms and used that
$$\rho_{V_{k+1}/V_{k-1}}(V_{k}) = \log\left(\frac{d(g_{\ast}\nu_{V_{k+1}/V_{k-1}})}{d(g_{\ast}m_{V_{k+1}/V_{k-1}})}(gV_k)\right).$$
In the fourth equality we disintegrated 
the integral with respect to $\nu_{(k-1,k,k+1)}$ 
using \eqref{eq:distintegration}, 
and in the last equality we made the change of variables $V_k\mapsto gV_k$ in the integral over $\mathbb{P}(V_{k+1}/V_{k-1})$, 
and recognized the integral over $\mathbb{P}(gV_{k+1}/gV_{k-1})$ as a KL divergence. 
Recall if $\mu\ll\nu$ are two measures on a measure space $\Omega$, then 
\[
  KL(\mu | \nu) = \int_{\Omega}\log\!\left(\frac{d\mu}{d\nu}\right) d\mu.
\]
\end{proof}

Now we can prove Theorem \ref{thm:Fisher-gaps}

\subsection{Proof of Theorem \ref{thm:Fisher-gaps}}
Since $d$ must be even, we replace it with $2d$ in what follows. 
We take $\mu$, its density $f:\Sp_{2d}(\mathbb{R})\to \mathbb{R}$, 
and $\lambda_1,...,\lambda_{2d}$ as in the statement of Theorem \ref{thm:Fisher-gaps}.
Further, we take $\nu_{V/U}$ and $\nu_{(k-1,k+1)}$ as in 
Proposition \ref{prp:KL-gap-formula}.

The main difficulty with obtaining a useful bound from
Proposition \ref{prp:KL-gap-formula} 
is the presence of the measures $\nu_{V/U}$ and $\nu_{(k-1,k+1)}$. 
These are, by nature, implicit objects, and 
are difficult to estimate. Our approach 
is to fix $V_{k-1}$ and $V_{k+1}$
and lower bound the integrand in Proposition \ref{prp:KL-gap-formula} 
uniformly for any fixed 
$V_{k-1}$ and $V_{k+1}$, relying only on the expectation 
over $g$. Indeed,
by Jensen's inequality
\[
  KL(\mu | \nu)\geq 0,
\]
with equality if and only if $\mu = \nu$. Hence, 
we have 
$$\mathbb{E}_g KL(g_{\ast}\nu_{V/U} | \nu_{gV/gU})\geq 0,$$
with equality if and only if $g_{\ast}\nu_{V/U} = \nu_{gV/gU},$
for $\mu$-almost every $g$. 

Intuitively, if the law of $\mu$ 
is spread over a large region in $\Sp_{2d}(\mathbb{R})$, 
then no family of measures $(\nu_{V/U})_{(U,V)\in \mathcal F_{(k-1,k+1)}}$ can be invariant 
under all elements in the support of $\mu$. 
This is made quantitative below, 
measuring how ``spread out''  $\mu$ is by the $L^1$ norm of 
its directional derivatives. 

\begin{proposition}
\label{prp:entropy-lower-bound}
Let $\mu$ be a probability distribution on
\(\Sp_{2d}(\mathbb{R})\), with density \(f\) with respect to Haar measure.
Let \((\nu_{V/U})_{(U,V)\in \mathcal F_{(k-1,k+1)}}\) be a family of
probability measures on the projective spaces \(\mathbb{P}(V/U)\), and assume that each
\(\nu_{V/U}\) has a density with respect to the standard
rotation-invariant measure on \(\mathbb{P}(V/U)\).  Then for every
\((U,V)\in \mathcal F_{(k-1,k+1)}\),
\[
\mathbb{E}_{g}KL(\nu_{V/U}\mid (g^{-1})_{\ast}\nu_{gV/gU})
\geq
c\left(
1+
\sup_{\substack{M\in \mathfrak{sp}_{2d}\\ 
\ \|M\|_{HS}=1}}
\|X_Mf\|_{L^1(\Sp_{2d}(\mathbb{R}))}
\right)^{-2}.
\]
\end{proposition}

\begin{proof}[Proof of Proposition \ref{prp:entropy-lower-bound}]
Let
\[
  A:=
  1+
  \sup_{\substack{M\in \mathfrak{sp}_{2d}\\ 
\ \|M\|_{HS}=1}}
  \|X_Mf\|_{L^1(\Sp_{2d}(\mathbb{R}))}.
\]

If \(A=\infty\), the desired lower bound is trivial.  
Otherwise, fix
\((U,V)\in\mathcal F_{(k-1,k+1)}\). 
By a regularization argument, we can assume that 
for each $g\in supp(\mu)$, the measure $(g^{-1})_{\ast}\nu_{gV/gU}$, 
on $\mathbb{P}(V/U)$, is 
absolutely continuous with respect to $\nu_{V/U}$ 
and has a positive density with respect to 
$m_{V/U}$. 

Recall Pinsker's inequality 
which says that for probability measures $\mu\ll\nu$, 
on a probability space $\Omega$ one has
\[
  KL(\mu|\nu)
  \geq
  \frac12\|\mu-\nu\|_{TV}^2.
\]
If additionally, $\nu\ll\mu$, the TV-norm can be written 
as 
\[
  \|\mu-\nu\|_{TV}
  =
  \int_\Omega\left|\frac{d\nu}{d\mu}-1\right|\,d\mu.
\]
Hence, applying this and Jensen's inequality gives
\[
\begin{aligned}
\mathbb E_g KL(\nu_{V/U}\mid (g^{-1})_*\nu_{gV/gU})
&\geq
\frac12\mathbb E_g\|\nu_{V/U}-(g^{-1})_*\nu_{gV/gU}\|_{TV}^2 \\
&\geq
\frac12
\left(
\int_{\mathbb P(V/U)}
\mathbb E_g\left|
\frac{d((g^{-1})_*\nu_{gV/gU})}{d\nu_{V/U}}(Y)-1
\right|\nu_{V/U}(dY)
\right)^2 .
\end{aligned}
\]
It is therefore enough to lower bound, uniformly in
\(Y\in\mathbb P(V/U)\), the quantity
\[
  \mathbb E_g\left|
  \frac{d((g^{-1})_*\nu_{gV/gU})}{d\nu_{V/U}}(Y)-1
  \right|
\]
by $A^{-1}$. For this we use the following two claims. 
The first is a non-concentration estimate that we will apply 
to $\log\left(d((g^{-1})_{\ast}\nu_{gV/gU})/d\nu_{V/U}(Y)\right)$.

\begin{claim}
Let \(M\in\mathfrak{sp}_{2d}\) and 
\(h:\Sp_{2d}(\mathbb R)\to\mathbb R\) be a measurable 
function such that
\begin{equation}
\label{eq:linear-movement}
 h(g\exp(tM))=h(g)+\alpha t,
\end{equation}
for some \(|\alpha|>0\).
Then for every \(\delta>0\),
\[
  \mathbb P_g(|h(g)|\leq\delta)
  \leq
  C|\alpha|^{-1}\delta\,
  \|X_Mf\|_{L^1(\Sp_{2d}(\mathbb R))}.
\]
\end{claim}

\begin{proof}[Proof of the claim]
Recall that $g$ has density $f(g)dm$, 
where \(m\) is Haar measure on \(\Sp_{2d}(\mathbb R)\). 
For any bounded smooth function \(\phi:\mathbb R\to\mathbb R\) with
bounded derivative, \eqref{eq:linear-movement} and integration 
by parts gives 
\begin{align*}
\mathbb{E}_{g}\phi'(h(g))
& = \int_{\Sp_{2d}(\mathbb{R})} \phi'(h(g))f(g)\,dm(g)\\
&=\frac1\alpha\int_{\Sp_{2d}(\mathbb{R})}
  \left.\frac{\partial}{\partial t}\right|_{t=0}
  \phi(h(g\exp(tM)))f(g)\,dm(g)\\
&=\frac1\alpha
  \left.\frac{\partial}{\partial t}\right|_{t=0}
  \int_{\Sp_{2d}(\mathbb{R})}\phi(h(g))f(g\exp(-tM))\,dm(g)\\
&=-\frac1\alpha\int_{\Sp_{2d}(\mathbb{R})}\phi(h(g))X_Mf(g)\,dm(g).
\end{align*}
Applying this with \(\phi\) such that
\(\phi'\geq\mathbf 1_{[-\delta,\delta]}\) and
\(\|\phi\|_\infty\leq C\delta\) gives
\[
\begin{aligned}
\mathbb P_g(|h(g)|\leq\delta)
&\leq\mathbb E_g\phi'(h(g))\\
&\leq |\alpha|^{-1}\|\phi\|_\infty\|X_Mf\|_{L^1\Sp_{2d}(\mathbb{R})}\\
&\leq C\delta|\alpha|^{-1}\|X_Mf\|_{L^1(\Sp_{2d}(\mathbb{R}))}.
\end{aligned}
\]
\end{proof}

To apply this, define any $Y \in \mathbb{P}(V/U)$, 
$$h_Y(g):=\log \frac{d((g^{-1})_*\nu_{gV/gU})}{d\nu_{V/U}}(Y).$$
We claim for each $Y$, there exists an $M\in \mathfrak{sp}_{2d}$ 
with $\|M\|_{HS} = 1$
such that $h_{Y}(g\exp(tM))$ changes linearly at a constant rate. 
This is done by taking $M$ to fix the flag $(U,Y,V)$ 
and generate a hyperbolic action on $V/U \cong \mathbb{R}^{2}$
expanding in the $Y$ direction.
We construct such an element in the following claim. 

\begin{claim}
For any $(U,Y,V)\in \mathcal F_{(k-1,k,k+1)}$,
there exists $M\in \mathfrak{sp}_{2d}$ and $|\alpha|\geq c$ such that 
$\|M\|_{HS} = 1$ and for all $g\in \Sp_{2d}(\mathbb{R})$
$$h_{Y}(g\exp(tM)) = h_Y(g) + \alpha t.$$ 
\end{claim}
\begin{proof}[Proof of the claim]
Up to conjugating by a symplectic rotation, we can assume
$(U,Y,V)$ is the standard flag, given by 
$U =\operatorname{span}(e_1,...,e_{k-1})$, 
$Y = U \oplus \operatorname{span}(e_{k})$, 
and $V = Y\oplus \operatorname{span}(e_{k+1})$ if $k<d$, while
$V = Y\oplus \operatorname{span}(f_d)$ if $k=d$. In both cases, 
take $M$ to be the rank-two
matrix
\[
M=\frac1{\sqrt2}\left(e_{k}e_{k}^{\ast}-f_{k}f_{k}^{\ast}\right).
\]
Here $vv^{\ast}$ is rank one projection onto 
$v$, so as a matrix in $\mathbb{R}^{2d\times 2d}$, 
$M$ has $\frac{1}{\sqrt{2}}$ in the $k,k$ entry and $\frac{-1}{\sqrt{2}}$ 
in the $d+k,d+k$ entry. Hence, $\exp(tM)$ 
preserves the subspaces $U$ and $V$. 
Moreover, its action on $V/U$ can be computed explicitly.
When $k<d$,  we identify $V/U$ with $\mathbb{R}^2$ 
in the basis $(e_{k},e_{k+1})$ and note 
the action of $\exp(tM)$ is given by 
$$
\exp(tM)|_{V/U}
=
\begin{pmatrix}e^{t/\sqrt2}&0\\0&1\end{pmatrix}.
$$
When $k=d$, we identify $V/U$ with $\mathbb{R}^2$ 
in the basis $(e_{d},f_d)$, and note 
we can write $\exp(tM)$ 
explicitly as
$$
\exp(tM)|_{V/U}
=
\begin{pmatrix}e^{t/\sqrt2}&0\\0&e^{-t/\sqrt2}\end{pmatrix},
$$
Using this, we compute
\[
\begin{aligned}
h_Y(g\exp(tM))
&=
\log\frac{d((\exp(-tM))_*((g^{-1})_{\ast}\nu_{gV/gU}))}{d\nu_{V/U}}(Y)\\
&=
h_Y(g)
+\log\frac{d((\exp(-tM))_*m_{V/U})}{dm_{V/U}}(Y)\\
& = h_Y(g) + \alpha t,
\end{aligned}
\]
for some $|\alpha|\geq c$. 
\end{proof}

Applying both claims with \(\delta=cA^{-1}\), with $c$ 
small enough, to
\(h_Y\) gives
\[
  \mathbb P_g(|h_Y(g)|\leq cA^{-1})
  \leq \frac{1}{2}
\]
Since
\(|e^x-1|\geq c\min(|x|,1)\), 
we obtain
\[
  \mathbb E_g\left|
  \frac{d((g^{-1})_*\nu_{gV/gU})}{d\nu_{V/U}}(Y)-1
  \right|
  =
  \mathbb E_g|e^{h_Y(g)}-1|
  \geq cA^{-1}.
\]
Substituting this into the Pinsker bound above proves the proposition.
\end{proof}

Now we use Proposition \ref{prp:KL-gap-formula}, Proposition
\ref{prp:entropy-lower-bound}, and 
invariance of KL divergence under pullback
to get
\begin{align*}
\lambda_{k}-\lambda_{k+1}
&=
\int_{\mathcal{F}_{(k-1,k+1)}}
\mathbb{E}_{g}KL(\nu_{V/U}\mid (g^{-1})_*\nu_{gV/gU})
\nu_{(k-1,k+1)}(d(U,V))\\
&\geq  
c\left(
  1+\sup_{\substack{M\in \mathfrak{sp}_{2d},\ \|M\|_{HS}=1}} \|X_Mf\|_{L^1(\Sp_{2d}(\mathbb R))}
  \right)^{-2}.
\end{align*}
The estimate now follows for each \(1\leq k\leq d\). The estimate 
for the remaining gaps follows
from the symmetry of the Lyapunov spectrum for symplectic matrices.

\section{Proof of Theorem \texorpdfstring{\ref{thm:top}}{Top Exponent}}
Throughout the section, we fix $E\in \mathbb{R}$, 
and 
assume that $A$ satisfies the following for some $K>0$. 
We let all constants depend on $E$ and $K$.

\begin{assumption}
\label{assumption:top-LE}
The following hold for some constant $K>0$:
\begin{enumerate}
    \item $A = A^T$ almost surely
    \item $(A_{xy})_{1\leq x\leq y \leq W}$ are independent 
    \item For all $x,y\in [1,W]$, $\mathbb{E} A_{xy} = 0$ and $0<\mathbb{E}A_{xy}^4 \leq K/W^2.$
    \item The variance profile is normalized in the sense that
    \[
      \max_{x\in[1,W]}\left|1-\sum_{y\in [1,W]} \Var(A_{xy})\right|\leq \frac{K}{W}.
    \]
\end{enumerate}
\end{assumption}

This is satisfied by the GOE law with $K = 12$, but also by more singular distributions,
such as centered Bernoulli random variables. 

\subsection{Basic Estimates}
In this section we collect a few basic estimates. 
First, we show the random matrix $A$ almost preserves the norms of 
vectors and randomizes their directions. 

\begin{lemma}\label{lem:conc}
If $A$ satisfies Assumption~\ref{assumption:top-LE}, 
then for any $u\in \R^W$,
$$\E\|Au\|^2 = \|u\|^2 + O(W^{-1}\|u\|^2),$$
and 
$$ \operatorname{Var}(\|Au\|^2)\leq \frac{C}{W}\|u\|^4.$$
Further, for any vectors $u,w\in \mathbb{R}^W$,
$$
  \E|\langle Au,w\rangle|^2
  \leq \frac{C}{W}\|u\|^2\|w\|^2.
$$
\end{lemma}

\begin{proof}
By scaling we can assume $u$ is a unit vector. 
Writing $u = (u_1,...,u_W)^T$ 
and using properties $(2), (3),$ and $(4)$,
of Assumption \ref{assumption:top-LE} we have
\begin{align*}
\mathbb{E}\|Au\|^2
& = \mathbb{E}\sum_{i}\left(\sum_{j}A_{ij}u_{j}\right)^2\\
& = \mathbb{E}\sum_{i} \sum_{j,k} A_{ij}A_{ik}u_{k}u_{j}\\
& = \sum_{j}\left(\sum_{i}\mathbb{E}A_{ij}^2\right) u_j^2\\
& = 1 + O(K/W).  
\end{align*}
Here and below all sums are over $[1,W]$. 
For the variance we estimate 
\begin{align*}
\operatorname{Var}\left(\|Au\|^2\right)
& = \sum_{i}\operatorname{Var}((Au)_i^2) + 2\sum_{i<j}\operatorname{Cov}((Au)_i^2, (Au)_j^2)\\
& = \sum_{i}\operatorname{Var}((Au)_i^2)  + 2\sum_{i<j} u_i^2u_j^2\operatorname{Var}(A_{ij}^2)\\
& \leq C\sum_{i}W^{-2} + 2C\sum_{i<j} u_i^2u_j^2 W^{-2}\\
& \leq \frac{C}{W}.
\end{align*}
We used properties (1), (2), and (3) of Assumption \ref{assumption:top-LE} 
to pass to the third line. 

To finish we consider $\langle Au,w\rangle$. 
Again we can assume $w = (w_1,...,w_W)$ is a unit vector, 
and let \(b_{ii}=u_iw_i\) and \(b_{ij}=u_jw_i+u_iw_j\) for \(i<j\). 
By independence and centering,
$$\mathbb{E}\left| \langle Au,w\rangle\right|^2 = \mathbb{E}\left|\sum_{i\leq j} b_{ij}A_{ij}\right|^2= \sum_{i\leq j} b_{ij}^2\Var(A_{ij})\leq \frac{C}{W}$$
since \(\sum_{i\leq j}b_{ij}^2\leq C\).
\end{proof}

We also use the following. 
\begin{lemma}
\label{lem:log-expectation}
Let $Z>0$ be a random variable such that $\E Z<\infty$. Then
\begin{equation}
\label{eq:log-identity}
  \E\log Z
  =
  \log\E Z
  + O\!\left(
    \frac{1}{\E Z}\E\frac{(Z-\E Z)^2}{Z}
  \right).
\end{equation}
\end{lemma}
\begin{proof}
For every $x>0$, the inequalities
$1-x^{-1}\leq\log x\leq x-1$ give
\[
  \left|\log x-(x-1)\right|
  =x-1-\log x
  \leq \frac{(x-1)^2}{x}.
\]
Apply this with $x=Z/\E Z$ and use
$\E(Z/\E Z-1)=0$.
\end{proof}

To finish, 
we record some elementary spectral properties of the matrix $M_E$, 
given by 
$$
  M_E=
  \begin{pmatrix}
    1+E^2 & 1 & 2E \\
    1 & 0 & 0 \\
    -E & 0 & -1
  \end{pmatrix}.
$$

\begin{lemma}[Spectral gap of $M_E$]
\label{lem:M-E-spectral}
For every $E\in\mathbb R$, the matrix $M_E$ has a unique maximal
eigenvalue \(\lambda_E\), satisfying
\[
\lambda_E\geq
\max\left(\frac{1+\sqrt5}{2},\frac{1+E^2}{2}\right),
\]
and its other two eigenvalues have absolute value at most $1$. Further, 
the right eigenvector corresponding to $\lambda_E$, normalized so that its first two coordinates sum to
$1$, is given by 
\[
v_E=
\begin{pmatrix}
\dfrac{\lambda_E}{\lambda_E+1}\\[2mm]
\dfrac{1}{\lambda_E+1}\\[2mm]
-\dfrac{E\lambda_E}{(\lambda_E+1)^2}
\end{pmatrix}.
\]
\end{lemma}

\begin{proof}
Let $P_E(\lambda)$ denote 
the characteristic polynomial of $M_E$, given by
\[
P_E(\lambda)
=
-\lambda^3+E^2\lambda^2-(E^2-2)\lambda+1
=
\bigl(-\lambda^3+2\lambda+1\bigr)
+E^2(\lambda^2-\lambda),
\]
and let \(\lambda_0=(1+\sqrt5)/2\).

First, we show that for any $E$, $P_E(\lambda)$ 
has a unique root in $(1,\infty)$.
For existence, note that for any 
$E$ we have $P_{E}(\lambda_0) = E^2(\lambda_0^2-\lambda_0)\geq E^2/2,$
and $P_{E}(\lambda)\to-\infty$ 
as $\lambda \to \infty$. Hence there is at least 
one root in $[1,\infty)$. 
For uniqueness of the root, 
note that 
for any $E$ and \(\lambda>1\), the equation $P_E(\lambda)=0$ is equivalent to
\[
E^2=
\frac{\lambda^3-2\lambda-1}{\lambda^2-\lambda}.
\]
For any fixed $E$ there is at most one $\lambda >1$,
satisfying that equation because the right-hand side is
a strictly increasing function of $\lambda$ on $(1,\infty)$. Indeed,
its derivative has numerator
\[
\lambda^4-2\lambda^3+2\lambda^2+2\lambda-1,
\]
which is positive at \(\lambda = 1\) and has positive derivative for
\(\lambda\geq1\).  
Hence for each $E$, there 
exists a uniquue root \(\lambda_E\in (1,\infty)\). 
The lower bound on $\lambda_E$ follows 
by noting that if \(s=E^2\geq1\), then
\[
P_E\left(\frac{1+s}{2}\right)
=\frac{(s-1)^3+16}{8}>0. 
\]
So \(\lambda_E\geq\max(\lambda_0,(1+E^2)/2)\).
See Figure 1 below.  

To see the other eigenvalues of $M_E$ have modulus 
at most $1$, 
note that \(\det M_E=1\) so
the product of the remaining eigenvalues is
\(\lambda_E^{-1}<1\). 
Hence, if they are complex, then each has modulus
\(\lambda_E^{-1/2}<1\). If they are real, neither can lie in
\((1,\infty)\) by uniqueness of \(\lambda_E\), 
and one can check that $P_E(\lambda)>0$ for all $\lambda<-1$. 
Indeed, both $P_0(\lambda)$ and
\(E^2(\lambda^2-\lambda)\) are positive.

Finally, to check that $M_Ev_E = \lambda_Ev_E$, 
note that solving $M_Ev=\lambda_Ev$ gives
\[
v=c\left(1,\lambda_E^{-1},-\frac{E}{\lambda_E+1}\right)^T,
\]
and choosing
\(c=\lambda_E/(\lambda_E+1)\) gives the stated normalization.
\end{proof}

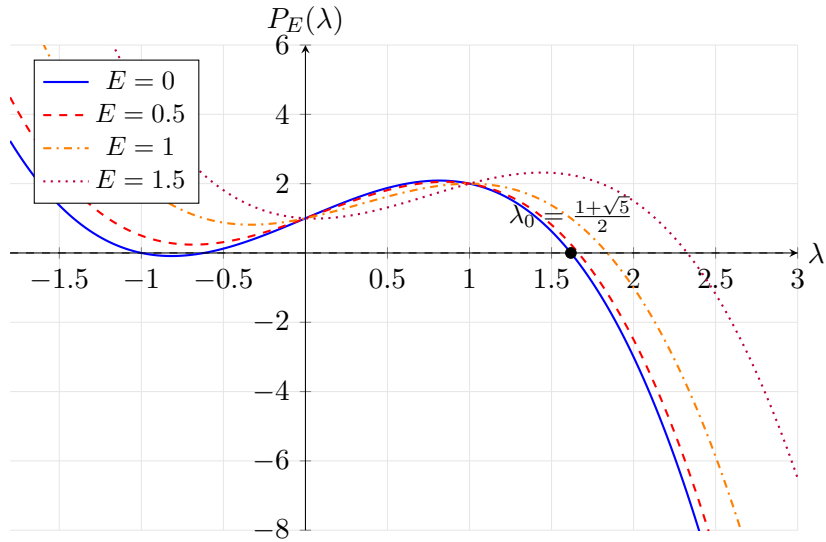
\begin{figure}[H]
\centering
\begin{tikzpicture}
\begin{axis}[
    xlabel={$\lambda$},
    ylabel={$P_E(\lambda)$},
    domain=-1.8:3,
    samples=200,
    axis lines=middle,
    ymin=-8, ymax=6,
    xmin=-1.8, xmax=3,
    width=12cm, height=8cm,
    legend pos=north west,
    legend style={font=\small},
    grid=both,
    grid style={line width=0.1pt, draw=gray!20},
    every axis x label/.style={at={(ticklabel* cs:1.0)}, anchor=west},
    every axis y label/.style={at={(ticklabel* cs:1.0)}, anchor=south},
]
\addplot[thick, blue] {-x^3 + 2*x + 1};
\addlegendentry{$E=0$}
\addplot[thick, red, dashed] {-x^3 + 0.25*x^2 - (0.25-2)*x + 1};
\addlegendentry{$E=0.5$}
\addplot[thick, orange, dashdotted] {-x^3 + x^2 - (1-2)*x + 1};
\addlegendentry{$E=1$}
\addplot[thick, purple, dotted] {-x^3 + 2.25*x^2 - (2.25-2)*x + 1};
\addlegendentry{$E=1.5$}
\addplot[thin, black, dashed] {0};
\addplot[only marks, mark=*, mark size=2pt, black] coordinates {(1.618, 0)};
\node[above] at (axis cs:1.618,0.3)
  {\small $\lambda_0=\frac{1+\sqrt5}{2}$};
\end{axis}
\end{tikzpicture}
\caption{The characteristic polynomial
$P_E(\lambda)=-\lambda^3+E^2\lambda^2-(E^2-2)\lambda+1$ for several values of
$E$. The unique root larger than $1$ moves to the right as $|E|$ increases.}
\label{fig:charpoly}
\end{figure}

\subsection{Estimates on the Stationary Measure}
The main input to Theorem \ref{thm:top} 
is the following estimate on the stationary measure.

\begin{proposition}\label{prop:stat}
For any $W\in \mathbb{N}$, let $\nu_{E,W}$ be the stationary measure on $\mathbb{P}(\R^{2W})$ associated to
the random matrix $T_E(A)$, given in \eqref{def:M}, 
and let $\lambda_E$ be the top eigenvalue of $M_E$, given in Lemma 
\ref{lem:M-E-spectral}. 
If $(v_1,v_2)^T\in \R^{2W}$ is a unit
vector with law $\nu_{E,W}$,
then 
\begin{align*}
  \E[\|v_1\|^2] &= \frac{\lambda_E}{1+\lambda_E} + O(W^{-1}),\\
  \E[\|v_2\|^2] &= \frac{1}{1+\lambda_E} + O(W^{-1}),\\
  \E[\langle v_1,v_2\rangle] &= \frac{-E\lambda_E}{(\lambda_E +1)^2} + O(W^{-1}).
\end{align*}
Furthermore, we have
$$
  \Var(\|v_1\|^2),
  \Var(\|v_2\|^2),
  \Var(\langle v_1,v_2\rangle)
  \le \frac{C}{W}.$$
\end{proposition}

\begin{remark}
  The values of $\mathbb{E}\|v_1\|^2$, etc, 
  are simply the coordinates of the top right eigenvector 
  of $M_E$, $v_E$ given in Lemma \ref{lem:M-E-spectral}.
\end{remark}

As noted in the proof sketch, 
the main observation is that the law
of $X:=(\|v_1\|^2, \|v_2\|^2, \langle v_1, v_2\rangle)$ 
is almost invariant invariant under the
induced projective action of $M_E$, 
given by 
\begin{equation}
\label{def:F-mapping}
F(x) := \frac{M_Ex}{(M_Ex)_1+(M_Ex)_2},
\end{equation}
where $(M_Ex)_i$ denotes the $i$-th entry 
of the vector $M_Ex\in \mathbb{R}^3$. 
Since $v_E$ 
is the unique most expanding direction 
of $M_E$, $F$ acts like a contraction with a fixed point at $v_E$. 
This lets us trap the law of $X$ near that fixed point 
via a Lyapunov, sometimes called a Margulis, functional. 

That the law of $X$ is almost invariant under $F$
is made rigorous in the following claim. 

\begin{lemma}
\label{lem:recurrence}
Let $X\in \mathbb{R}^3$ be the random vector $(\|v_1\|^2, \|v_2\|^2, \langle v_1, v_2\rangle)^T$,
where $(v_1, v_2)^T$ is a unit vector with law $\nu_{E,W}$. 
There exists a random variable $Y\in \R^3$ such that
\[
  X\stackrel{d}{=}
  F(X + M_E^{-1}Y),
\]
and for almost every realization of $X$
\[
  |\mathbb{E}[Y|X]|,
  \Var(Y|X) \leq \frac{C(1+E^2)}{W}.
\]
\end{lemma}

\begin{proof}
By stationarity and matrix multiplication,
\begin{align*}
X \stackrel{d}{=}
  \frac{1}{\|Av_1-Ev_1-v_2\|^2+\|v_1\|^2}
  \begin{pmatrix}
    \|Av_1-Ev_1-v_2\|^2\\
    \|v_1\|^2\\
    \langle Av_1-Ev_1-v_2,v_1\rangle
  \end{pmatrix}.
\end{align*}
Hence, if we define $Y= (Y_1,0,Y_3)^{T}$, with
\begin{align*}
  Y_1
  &:=
  \|Av_1\|^2-\|v_1\|^2-2\langle Av_1,v_2\rangle-2E\langle Av_1,v_1\rangle,\\
  Y_3
  &:=
  \langle Av_1,v_1\rangle,
\end{align*}
then algebra gives the relation
\begin{align*}
  \begin{pmatrix}
    \|Av_1-Ev_1-v_2\|^2\\
    \|v_1\|^2\\
    \langle Av_1-Ev_1-v_2,v_1\rangle
  \end{pmatrix}
  &=
  M_E X + Y.
\end{align*}
Normalizing by the sum of the first two coordinates, 
this implies the first claim.

For the conditional expectation of \(Y\), condition first on the full vector
\((v_1,v_2)\).  The linear terms in \(Y_1\) and \(Y_3\) have mean zero.  If
\(r_j=\sum_i\Var(A_{ij})\) and \(D=\operatorname{diag}(r_j)\), then
\(\E A^2=D\), and Assumption \ref{assumption:top-LE} gives
$$\E[Y_1\mid v_1,v_2]
  =
  \langle (D-\Id)v_1,v_1\rangle
  =
  O(W^{-1}\|v_1\|^2),$$
  and
  $$
  \E[Y_3\mid v_1,v_2]=0.
$$
Taking the conditional expectation over $\|v_1\|^2, \|v_2\|^2, \langle v_1,v_2\rangle,$
then gives 
$$|\E[Y\mid X]|\le C/W.$$  

For the conditional variance, Lemma~\ref{lem:conc},
implies that for any fixed $v_1,v_2$
\[
  \Var(Y\mid v_1,v_2)
  \leq
  \frac{C}{W}\left(\|v_1\|^4+\|v_1\|^2\|v_2\|^2
  +E^2\|v_1\|^4\right)
  \leq
  \frac{C(1+E^2)}{W}.
\]
Here the first term comes from
\(\|Av_1\|^2-\E[\|Av_1\|^2\mid v_1]\), and the other terms from
\(\langle Av_1,v_2\rangle\) and \(\langle Av_1,v_1\rangle\).
The law of total variance then gives
\[
  \Var(Y\mid X)
  \leq
  \E[\Var(Y\mid v_1,v_2)\mid X]
  +
  \Var(\E[Y\mid v_1,v_2]\mid X)
  \leq
  \frac{C(1+E^2)}{W},
\]
since \(\E[Y\mid v_1,v_2]\) is \(O(W^{-1})\).
\end{proof}

Now we prove the proposition. 

\begin{proof}[Proof of Proposition~\ref{prop:stat}]
Let \(\lambda_E\) and \(v_E\) be as in Lemma~\ref{lem:M-E-spectral}.
We argue that $F$ is a mapping on the space 
$$ S:=\{(x_1,x_2,x_3):x_1,x_2\geq0,\ x_1+x_2=1,\
  |x_3|\leq \sqrt{x_1x_2}\},$$
which has $v_E\in S$ as an attracting fixed point, 
and thus Lemma 3.7 implies $X$, which takes values in $S$
must concentrate its law near $v_E$. 

To implement this we introduce a function 
measuring distance 
to $v_E$, which contracts under $F$. 
Define $d:S\to \mathbb{R}_{\geq 0}$ given by 
\begin{equation}
\label{def:distance-metric}
d(x):= \sup_{k\geq 0} r_0^{k}\|M_E^k(\frac{x}{\phi_E\cdot x} - v_E)\|,
\end{equation}
where $r_0\in (\lambda_{E}^{-1},1)$ is any constant, 
$\phi_E$ is the top left eigenvector 
of $M_E$, having eigenvalues $\lambda_E$, given by 
$$\phi_E:= c_0(1, \lambda_E^{-1}, \frac{2E}{\lambda_E + 1})^T,$$
where $c_0>0$ is chosen so that $\phi_E\cdot v_E = 1$.

Briefly we note that that $d$ exists and is Lipschitz. 
Indeed, one can check that 
$\phi_E\cdot x$ is bounded uniformly away from 
$0$ for all $x\in S$, and so 
$\frac{x}{\phi_E\cdot x}-v_E$, 
exists for all $x\in S$, 
and is orthogonal to $\phi_E$. 
The space of vectors orthogonal to $\phi_E$ is preserved by $M_E$, 
and by Lemma \ref{lem:M-E-spectral}, there exist $C>0$ such that
for all $k$,
$r_0^kM_E^k$ has operator norm at most $C$ on that subspace.
This ensure $d(x)$ exists for all $x\in S$, but also implies 
$d$ is Lipschitz, since 
it is a supremum of Lipschitz maps with Lipschitz constants
bounded uniformly by $C$. 

The following claim collects the important properties 
of $d$. 

\begin{claim}
\label{lem:action-of-F}
There exists constants $C>0$ and $\rho\in (0,1)$, all depending on $E$, 
such that 
\begin{enumerate}
\item for all $x\in S$ 
$$\|x-v_E\|\leq Cd(x),$$
\item for all $x\in S,$
$$d(F(x)) \leq \rho d(x),$$
\item for all $x\in S$ and $y\in \mathbb{R}^3$ with $F(x+y)\in S$: 
$$|d(F(x+y)) - d(F(x))|\leq C\|y\|.$$
\item for all $x\in S$ and $y\in \mathbb{R}^3$ with $F(x+y)\in S$:
$$\|F(x+y)-F(x)-DF(x)y\|\leq C\|y\|^2.$$
\end{enumerate}
Recall $v_E$ is the eigenvector of $M_E$ corresponding to the largest eigenvalue. 
\end{claim}

\begin{proof}[Proof of claim]
For property $(1)$, we estimate 
\begin{align*}
\|x-v_E\| 
& = \|\phi_E\cdot x(\frac{x}{\phi_E\cdot x}-v_E + v_E) - v_E\|\\
& \leq | \phi_E\cdot x| \|\frac{x}{\phi_E\cdot x}-v_E\| + |\phi_E\cdot x-1|\\
& \leq C\|\frac{x}{\phi_E\cdot x}-v_E\|\\
& \leq Cd(x), 
\end{align*}
where the last inequality comes from the $k=0$ term in the definition of $d$.

For property $(2)$, compute 
\begin{align*}
d(F(x))
& = \sup_{k\geq 0}r_0^k\|M_E^k(\frac{F(x)}{\phi_E\cdot F(x)} - v_E)\|\\
& = \sup_{k\geq 0} r_0^k \|M_E^k(\frac{M_E x}{\lambda_E \phi_E\cdot x} - v_E)\|\\
&  = \sup_{k\geq 0} r_0^k \lambda_E^{-1}\|M_E^{k+1}(\frac{x}{\phi_E \cdot x} - v_E)\|\\
& \leq r_0^{-1}\lambda_E^{-1} d(x).
\end{align*}
Since we took $r_0\in (\lambda^{-1}, 1)$, we have $r_0^{-1}\lambda_E^{-1}< 1$. 

We next prove property $(4)$. Since $(M_Ex)_1+(M_Ex)_2$ is bounded
uniformly away from zero on the compact set $S$, the map $F$ has uniformly
bounded second derivatives in a fixed neighborhood of $S$. Thus, for
$\|y\|$ sufficiently small, Taylor's theorem gives
\[
\|F(x+y)-F(x)-DF(x)y\|\leq C\|y\|^2.
\]
If $\|y\|$ is bounded away from zero, the same estimate follows from the
boundedness of $F(x)$, $F(x+y)$, and $DF(x)$, since
$1+\|y\|\leq C\|y\|^2$.

Finally, property $(3)$ follows quickly from property $(4)$. If $\|y\|\leq 1$,
then the Lipschitz property of $d$ and the boundedness of $DF$ on $S$ give
\[
|d(F(x+y))-d(F(x))|
\leq C\|F(x+y)-F(x)\|
\leq C\|y\|.
\]
If $\|y\|\geq 1$, the same bound follows from the boundedness of $d$ on $S$.
\end{proof}

Using $d$ we can estimate the law of the stationary measure. 
Applying Lemma \ref{lem:recurrence}, 
and Claim \ref{lem:action-of-F}, and 
we have for $\rho\in (0,1)$ as in Claim \ref{lem:action-of-F}, and any $\epsilon>0$
\begin{align*}
\mathbb{E}d(X)^2 
& = \mathbb{E}d\left(F(X + M_E^{-1}Y)\right)^2\\
& = \mathbb{E} d\left(F(X) + (F(X + M_E^{-1}Y) - F(X))\right)^2\\
& \leq (1+\epsilon)\mathbb{E} d(F(X))^2 + C\epsilon^{-1}\mathbb{E}\|F(X + M_E^{-1}Y) - F(X)\|^2\\
& \leq (1+\epsilon)\rho^2\mathbb{E}d(X)^2 + \frac{C}{\epsilon W}.
\end{align*}
Taking $\epsilon$ small enough 
that $\rho^2(1+\epsilon)<\rho$, 
and re-arranging gives
$$\mathbb{E}d(x)^2 \leq \frac{C}{(1-\rho)W}.$$
Thus using Claim \ref{lem:action-of-F} again gives  
$$\mathbb{E}\|X-v_E\|^2 \leq C\mathbb{E}d(X)^2 \leq C/W.$$
which implies the variance estimates.

To estimate the mean, note that the above only implies 
$$\|\mathbb{E}X - v_E\| = O(W^{-1/2}).$$
To improve the error term, we estimate
\begin{align*}
\mathbb{E}X
- F\left(\mathbb{E}X\right)
& = \mathbb{E}F\left(X + M_E^{-1}Y\right) - F\left(\mathbb{E}X\right)\\
& = \mathbb{E}\left(DF(\mathbb{E}X)(X-\mathbb{E}X + M_E^{-1}Y) + O(\|X-\mathbb{E}X + M_E^{-1}Y\|^2)\right)\\
& = O(W^{-1}).
\end{align*} 
The first line followed from Lemma \ref{lem:recurrence}, 
the second followed from property $(4)$ of Claim \ref{lem:action-of-F}.
The third line follows since 
$$\|\mathbb{E}(X+ M_E^{-1}Y -\mathbb{E}X)\| = \|M_{E}^{-1}\mathbb{E}Y\| = O(W^{-1}),$$
and 
$$ \mathbb{E}\|X+ M_E^{-1}Y -\mathbb{E}X\|^2 \leq C\mathbb{E}\|X-\mathbb{E}X\|^2 + C\mathbb{E}\|M_E^{-1}Y\|^2=  O(W^{-1}).$$

Finally, since $\mathbb{E}X\in S$, and $d$ is Lipshitz,
we thus have
$$d(\mathbb{E}X) \leq d(F(\mathbb{E}X)) + C\|\mathbb{E}X-F(\mathbb{E}X)\|\leq \rho d(\mathbb{E}X) + \frac{C}{W}.$$
Re-arranging and applying property $(1)$ in Claim \ref{lem:action-of-F}
implies the result. 
\end{proof}

\subsection{Completion of the Proof of Theorem \texorpdfstring{\ref{thm:top}}{Top Exponent}}
Let $T_E(A)$ be the matrix in \eqref{def:M}. 
If $A$ satisfies Assumption \ref{assumption:top-LE}, 
the hypothesis Furstenberg's Theorem is satisfied, 
and so we have 
$$L_1(E,W) = \frac{1}{2}\mathbb{E}_{A}\mathbb{E}_{v} \log \|T_E(A)v\|^2,$$
where we have written $\mathbb{E}_{v}$ 
for the integral over $\mathbb{P}(\mathbb{R}^{2W})$
with respect to the measure $\nu_{E,W}$. 
We now resolve the expectations, 
using Lemma \ref{lem:conc} and Proposition \ref{prop:stat}. 

First we estimate the expectation over $A$. For any 
unit vector $v:= (v_1,v_2)^{T}\in \mathbb{R}^{2W}$ with $v_1\neq 0$, 
Lemma \ref{lem:log-expectation} implies
\begin{align*}
&\left|
  \E_A\log\|T_E(A)v\|^2
  -\log\E_A\|T_E(A)v\|^2
\right|\\
&\leq
  \frac{1}{\E_A\|T_E(A)v\|^2}
  \E_A\frac{
    \left(\|T_E(A)v\|^2-\E_A\|T_E(A)v\|^2\right)^2
  }{\|T_E(A)v\|^2}\\
&\leq
  \frac{C}{\|v_1\|^2}
  \left(
    \Var_A(\|Av_1\|^2)
    +\E_A|\langle Av_1,Ev_1+v_2\rangle|^2
  \right)\\
&\leq
  \frac{C}{W}
  \left(\|v_1\|^2+\|Ev_1+v_2\|^2\right)\\
&\leq \frac{C}{W}.
\end{align*}
The first inequality is Lemma~\ref{lem:log-expectation}. For the second, we
used that $\|T_E(A)v\|^2\geq\|v_1\|^2$, 
$c\leq \E_A\|T_E(A)v\|^2 \leq C$, and the identity
\[
  \|T_E(A)v\|^2-\E_A\|T_E(A)v\|^2
  =\|Av_1\|^2-\E_A\|Av_1\|^2
  -2\langle Av_1,Ev_1+v_2\rangle.
\]
The third inequality follows from Lemma~\ref{lem:conc},
and the last follows from $\|v\|=1$.
If $v_1 = 0$, the RHS can be replace by $0$ since $T_E(A)v$ 
is then deterministic. 

If 
\[
  Q_E(v):=(2+E^2)\|v_1\|^2+\|v_2\|^2
  +2E\langle v_1,v_2\rangle,
\]
then Lemma~\ref{lem:conc} gives
\[
  \E_A\|T_E(A)v\|^2=Q_E(v)+O(W^{-1}),
\]
so 
$$L_1(E,W) = \frac{1}{2}\mathbb{E}_{v}\log Q_E(v) + O(W^{-1}).$$
The expectation over $v$ distributed according 
to $\nu_{E,W}$, 
and since $Q_E(v)$ is bounded 
uniformly above and below, we can apply Lemma 3.3 again, 
and the stationary measure estimate Proposition \ref{prop:stat}
to obtain
\begin{align*}
L_1(E,W)
& = \frac{1}{2}\mathbb{E}_{v}\log Q_E(v) + O(W^{-1})\\
& = \frac{1}{2} \log \mathbb{E}_{v}Q_E(v)
  + O\left(\frac{1}{\mathbb{E}_{v}Q_E(v)}
  \mathbb{E}_{v}\frac{(Q_E(v)-\mathbb{E}_{v}Q_E(v))^2}{Q_E(v)}\right)
  + O(W^{-1})\\
& = \frac{1}{2}\log \mathbb{E}_{v}Q_E(v) + O(\operatorname{Var}(Q_E(v)))+ O(W^{-1})\\
& = \frac{1}{2}\log \lambda_E + O(W^{-1}). 
\end{align*} 
The third equality follows from the same uniform bounds. 
The last equality followed by estimating 
the variance of $Q_E(v)$ using Proposition \ref{prop:stat}, 
and using the observation that 
$\mathbb{E}Q_E(v) = \lambda_E+O(W^{-1}).$
This proves the Theorem.

\section{Malliavin Calculus for Transfer Matrices}
\label{sec:Malliavin-estimates}
Throughout the section, we fix $E\in\mathbb R$ and take $A$ to be a
normalized GOE matrix. Thus $A$ is a centered Gaussian vector in
$\mathbb R_{\mathrm{Sym}}^{W\times W}$ with
$$
  \operatorname{Cov}(A)=\frac{2}{W}\Id
$$
relative to the Hilbert--Schmidt inner product. Equivalently, 
$A = A^T$ and $(A_{xy})_{1\leq x\leq y \leq W}$ are 
indepednent $\mathcal N(0,\frac{1+\delta_{xy}}{W})$, 
variables. All constants in this
section may depend on $E$.

The main result of this section is the following.
\begin{proposition}
\label{prp:Fisher-bound}
Fix $E\in\mathbb R$. There exists $C:=C(E)$ such that,
if $W$ is sufficiently large, 
$A$ is a normalized GOE matrix and $\mu$ is the law of $T_E(A)$, 
given in \eqref{def:M},
then $\mu^{\ast 1000}$ is absolutely continuous with respect 
to Haar measure $dm$ on $\Sp_{2W}(\mathbb{R})$, 
and
$$
  \sup_{\substack{M\in \mathfrak{sp}_{2W}\\ \|M\|_{HS}=1}}
  \|X_M\frac{d\mu^{*1000}}{dm}\|_{L^1(\Sp_{2W}(\mathbb R))}^2\leq CW .
$$
\end{proposition}

We emphasize that by dimension counting, 
none of 
$\mu, \mu^{\ast 2}, \mu^{\ast 3}$ can be 
absolutely continuous with respect to Haar measure. 
In this sense, 
Proposition \ref{prp:Fisher-bound} 
can be interpreted as a quantitative 
hypoellipticity estimate.
Indeed, $\mu^{\ast 1000}$ is the law of the random variable 
\begin{equation}
F_E(A_1,...,A_{1000}) := T_E(A_1)T_E(A_2)\cdots T_E(A_{1000}).
\end{equation}
$\Sp_{2W}(\mathbb{R})$ has dimension $2W^2 + W$, 
but wiggling any given $A_i$ above only moves $F_E(A_1,...,A_{1000})$
on a $W(W+1)/2$ dimensional submanifold. 
Hence for $\mu^{\ast 1000}$ to even have a density, much less one with a gradient in $L^1$, 
one needs to exploit the randomness coming from many $A_i$, and, in particular, 
verify that each $A_i$ wiggles $F_E$ in directions which are transverse to 
those generated by other $A_j$. A similar problem is encountered 
in degenerate diffusions and the theory of hypoellipticity 
\cite{Hormander1967,KusuokaStroock1985,Hairer2011,CoopermanRowan2025}.
In our setting, getting quantitative estimates which are sharp in $W$ 
is an added difficulty.

\subsection{Conventions and Basic Facts}
\label{sec:malliavin-conventions}
Here we collect some definitions and basic facts. 
Define the vector space $\mathcal H_W$ by
\begin{equation}
\label{def:H-sub-W}
  \mathcal{H}_W:=(\mathbb{R}_{\mathrm{Sym}}^{W\times W})^{1000},
\end{equation}
and endow it with the 
standard inner product structure that makes it
isometric to $\mathbb{R}^{1000W(W+1)/2}$ with the standard Euclidean inner product. 

All calculations with symmetric matrices $A_i$
could be done in the standard Euclidean space $\mathbb{R}^{W(W+1)/2}$, 
however, we find it more convenient to work with the symmetric matrix representation. 
In this representation, the inner product $\langle \cdot, \cdot \rangle_{\mathcal H_W}$, 
of $(H_1,\cdots, H_{1000}), (K_1,\cdots, K_{1000})\in \mathcal H_W$, is given by 
\begin{equation}
\label{def:inner-product-H}
\langle (H_1,\ldots,H_{1000}), (K_1,\ldots,K_{1000})\rangle_{\mathcal H_W}
= \sum_{i=1}^{1000} \operatorname{Tr}(H_iK_i).
\end{equation}
Divergences of vector fields or gradients of functions on $\mathcal H_W$,
are defined using \eqref{def:inner-product-H}, 
and the standard identification between a vector space and 
its tangent space. For example, if
$X:\mathcal H_W\to \mathcal H_W$ is a smooth vector field, then
for each $A\in \mathcal H_W$ we view $DX(A):\mathcal H_W\to\mathcal H_W$
as a linear mapping from $\mathcal H_W$ to $\mathcal H_W$
and define the divergence of $X$, $\operatorname{div} X:\mathcal H_W\to \mathbb{R}$ by
\begin{equation}
\label{def:divergence}
  \left(\operatorname{div}X\right)(A)=\operatorname{Tr}(DX(A))
  =
\sum_\alpha \langle DX(A)e_\alpha,e_\alpha\rangle_{\mathcal H_W},
\end{equation}
where $(e_\alpha)$ is an orthonormal basis for $\mathcal H_W$. 
Similarly, adjoints, $\ell^2$ norms, and Hilbert--Schmidt norms on
$\mathcal H_W$ are induced by \eqref{def:inner-product-H}. For example 
given any linear mapping $L:\mathcal H_W\to \mathcal H_W$, 
its Hilbert Schmidt norm is given by 
\begin{equation}
\label{def:HS-norm}
\|L\|_{HS(\mathcal H_W)}^2
= \operatorname{Tr}(LL^T) = 
\sum_\alpha \|Le_\alpha\|_{\mathcal H_W}^2.
\end{equation}

We also use some basic facts about $\Sp_{2W}(\mathbb{R})$ and 
its Lie algebra, see for example
\cite{Hall2015}. 
First, we can write the Lie algebra 
of
$\Sp_{2W}(\mathbb R)$, which is just the tangent space at the identity, as
the vector space
\begin{equation}
\label{def:Lie-Algebra-representation}
\mathfrak{sp}_{2W}
  =
  \left\{
  \begin{pmatrix}
    P & Q\\
    R & -P^T
  \end{pmatrix}
  : P,Q,R\in \mathbb{R}^{W\times W},\ Q=Q^T,\ R=R^T
  \right\},
\end{equation}
equipped with the standard inner product
\begin{equation}
\label{def:lie-algebra-inner-product}
\langle M,N\rangle_{\mathfrak{sp}_{2W}}
  :=
\operatorname{Tr}(M^TN). 
\end{equation}
Further for any $g\in \Sp_{2W}(\mathbb{R})$ we recall that  
$\operatorname{Ad}_{g}:\mathfrak{sp}_{2W}\to \mathfrak{sp}_{2W}$ is the adjoint action of $g$ on the Lie algebra,
given by 
\begin{equation}
\label{def:adjoint-action}
\operatorname{Ad}_g M = gMg^{-1}.
\end{equation}
Finally, we let $L_g:\Sp_{2W}(\mathbb{R})\to \Sp_{2W}(\mathbb{R})$ denote left multiplication by $g$, 
and $d(L_g)_h:T_{h}\Sp_{2W}(\mathbb{R})\to T_{gh}\Sp_{2W}(\mathbb{R})$ 
be the derivative of $L_g$ at $h$. 
Adjoints and operator norms involving $\mathfrak{sp}_{2W}$ 
are always defined relative to the inner product above. 

\subsection{Lifting Vector Fields}
We will use the following lemma which essentially follows
by Gaussian integration by parts. Given 
a Gaussian $Z\in \mathbb{R}^d$, 
it lets us estimate the derivatives 
of the density of the random variable $F(Z)\in \Sp_{2W}(\mathbb{R})$
in some direction $X$
by the derivatives of the Gaussian density of $Z$
in some direction $V$. 
The equation
\eqref{def:lifting-relation} below 
means the flow generated by \(V\) in $\mathbb{R}^d$ 
is mapped to the flow of \(X\) in $\Sp_{2W}(\mathbb{R})$ 
after applying \(F\).

\begin{lemma}
\label{lem:abstract-gaussian-lifting}
Let $W, d\in \mathbb{N}$, \(Z\in\mathbb{R}^d\) be a centered 
Gaussian vector with positive definite covariance $\Sigma$,
$F:\R^d\longrightarrow \Sp_{2W}(\R)$
be a smooth mapping, and $\mu$ be the law of 
\(F(Z)\) on $\Sp_{2W}(\mathbb{R})$. 
Then for any left-invariant vector field \(X\) on \(\Sp_{2W}(\R)\)
and \(V\in W^{1,1}_{loc}(\R^d;\R^d)\) satisfying
\begin{equation}
\label{def:lifting-relation}
DF(z)V(z)=X(F(z))
\end{equation}
for a.e. $z\in \mathbb{R}^d$, we have
\begin{equation}
\label{ineq:abstract-gaussian-lifting}
\|X\mu\|_{TV(\Sp_{2W}(\R))}
\leq
\left(
\E\left\|\Sigma^{-1/2}V(Z)\right\|_{2}^2
+
\E\left\|
  \Sigma^{-1/2}DV(Z)\Sigma^{1/2}
\right\|_{\mathrm{HS}}^2
\right)^{1/2}.
\end{equation}
The LHS can be replaced by  $\|X\frac{d\mu}{dm}\|_{L^1(\Sp_{2W}(\R))},$
whenever $\mu$ is absolutely continuous with respect 
to the Haar measure $dm$. 
\end{lemma}

\begin{remark}
One can think of the right side of \eqref{ineq:abstract-gaussian-lifting}
as a Gaussian \(H^1\)-norm of \(V\).  
\end{remark}

We view $X\mu$ as a distribution on $\Sp_{2W}(\mathbb{R})$, and recall that the 
total variation norm on such distributions 
is given by 
$$\|\nu\|_{TV(\Sp_{2W}(\mathbb{R}))} = \sup_{\stackrel{\phi\in C^{\infty}_c(\Sp_{2W}(\mathbb{R}))}{\|\phi\|_{\infty}\leq 1}}\int_{\Sp_{2W}(\mathbb{R})}\phi(g) \nu(dg).$$

\begin{proof}
By a standard regularization argument, we can assume $V$ is smooth. 
Hence for any \(\phi\in C_c^\infty(\Sp_{2W}(\R))\),
we have
\begin{align*}
\langle X\mu,\phi \rangle
  & = -\int_{\Sp_{2W}(\R)}(X\phi)(g)\mu(dg)\\
  & = -\int_{\mathbb{R}^d}(X\phi)(F(z))\rho_\Sigma(z)\,dz\\
  & = -\int_{\mathbb{R}^d}(V\phi\circ F)(z)\rho_\Sigma(z)\,dz\\
  & = \int_{\mathbb{R}^d}(\phi\circ F)(z)\operatorname{div}(\rho_\Sigma V)(z)\,dz\\
  & \leq \|\operatorname{div}(\rho_{\Sigma}V)\|_{L^1(\mathbb{R}^d)}\|\phi\|_{L^{\infty}(\Sp_{2W}(\mathbb{R}))},
\end{align*}
where $\rho_{\Sigma}$ denotes the density of $Z$. 
The first equality followed by the definition of the distribution $X\mu$. 
The second equality followed by the definition 
of $\mu$ and $\rho_{\Sigma}$. The third equality 
by \eqref{def:lifting-relation} 
and its implication that 
\[
  (V\phi\circ F)(z) = (X\phi)(F(z))
\]
for almost every \(z\),  
and the fourth by standard integration by parts on $\mathbb{R}^d$. 
Taking the supremum over $\|\phi\|_{L^{\infty}(\Sp_{2W}(\mathbb{R}))} \leq 1$,
we have 
\[
  \|X\mu\|_{TV(\Sp_{2W}(\mathbb{R}))}
  \leq
  \|\operatorname{div}(\rho_{\Sigma}V)\|_{L^1(\mathbb{R}^d)}.
\]
It remains to estimate the RHS. 

By making the change 
of variables \(Z\to\Sigma^{-1/2}Z\) we can assume, without loss of 
generality that $Z$ is the standard Gaussian vector, 
so that $\Sigma = \operatorname{I}_d$, and we have
$$\rho_{\Sigma}(z) = \frac{1}{(2\pi)^{d/2}}e^{-\frac{1}{2}\|z\|^2}.$$
For such \(Z = (Z_1,\cdots, Z_d)\), Gaussian integration by parts implies
$$\E[Z_i h(Z)]=\E[\partial_i h(Z)].$$

Computing the divergence, we have
$$\|\operatorname{div}(\rho_\Sigma V)\|_{L^1(\mathbb{R}^d)} = \|\rho_{\Sigma}(\operatorname{div}V - \langle z, V\rangle)\|_{L^1(\mathbb{R}^d)} = \mathbb{E}|\operatorname{div}V - \langle z, V\rangle|.$$
Applying the above and Jensen's inequality, then Gaussian integration by parts gives
\[
\begin{aligned}
  \|\operatorname{div}(\rho_{\Sigma}V)\|^2_{L^1(\mathbb{R}^d)}
  & \leq \E|\operatorname{div}V(Z)-Z\cdot V(Z)|^2\\
  &=
  -\E\left[
    V(Z)\cdot \nabla\left(\operatorname{div}V(Z)-Z\cdot V(Z)\right)
  \right]\\
  &=
  \E|V(Z)|^2
  -
  \E\sum_{i,j=1}^d
  V_j(Z)\partial_j\partial_i V_i(Z)\\
  &\quad+
  \E\sum_{i,j=1}^d
  Z_iV_j(Z)\partial_j V_i(Z)\\
  &=
  \E|V(Z)|^2
  +
  \E\sum_{i,j=1}^d
  \partial_i V_j(Z)\partial_j V_i(Z)\\
  &\leq
  \E|V(Z)|^2+\E\|DV(Z)\|_{\mathrm{HS}}^2 .
\end{aligned}
\]
Note that by a regularization argument we can assume $V$ 
is twice differentiable in the above computation.  

Combining with the above estimate gives \eqref{ineq:abstract-gaussian-lifting}. 
If $\mu\ll dm$, then the conclusion follows since 
$$\|X\mu\|_{TV(\Sp_{2W}(\mathbb{R}))} = \|X\frac{d\mu}{dm}\|_{L^1(\Sp_{2W}(\mathbb{R}))}.$$
\end{proof}

\subsection{Estimates on Operators} 
Here we collect the estimates used in the proof 
of Proposition \ref{prp:Fisher-bound}. The reader is encouraged 
to look at the proof of Proposition \ref{prp:Fisher-bound} 
before reading this section. However, to motivate Proposition \ref{prp:operator-estimates}
below,
we remark that we will apply Lemma \ref{lem:abstract-gaussian-lifting} with $F=F_E$, 
defined in \eqref{eq:def-F-mapping}, $X=X_M$,
and the lifting vector field $V = \tilde{X}_M$ given by 
$$\tilde{X}_M(a) = DF_E(a)^{\ast}\Gamma(a)^{-1}X_M(F_E(a)).$$
Here $\Gamma(a)$ is the so called \textit{Malliavin covariance matrix} as defined in \eqref{def:Gamma}.
Therefore, to estimate $\|\tilde{X}_M(a)\|_{\mathcal H_W}$ and $\|D\tilde{X}_M(a)\|_{HS(\mathcal H_W)}$ 
we need some estimates on $\Gamma(a)^{-1}$ and $DF_E(a)$, and their derivatives, which is 
the content of Proposition \ref{prp:operator-estimates}.

We use the mappings $T_E:\mathbb R_{\mathrm{Sym}}^{W\times W}\to
\Sp_{2W}(\mathbb R)$ and $F_E:\mathcal H_W\to\Sp_{2W}(\mathbb R)$ given by
\begin{equation}
\label{def:T-mapping}
T_E(A):=
\begin{pmatrix}
  A-E\Id_W&-\Id_W\\
  \Id_W&0
\end{pmatrix},
\end{equation}
\begin{equation}
\label{eq:def-F-mapping}
F_E(a_1,\ldots,a_{1000}):=T_E(a_1)\cdots T_E(a_{1000}).
\end{equation}
It is more convenient to work with the left-trivialized 
versions of these operators and their derivatives, 
which we now introduce. 

Define
$\mathcal B(a):
  \mathcal H_W\to \mathfrak{sp}_{2W}$, given by 
\begin{equation}
\label{def:B-of-a}
\mathcal B(a):=
  d(L_{F_E(a)^{-1}})_{F_E(a)}DF_E(a)
\end{equation}
and $\widehat{\Gamma}(a):=
  \mathfrak{sp}_{2W}\to \mathfrak{sp}_{2W},$ given by
\begin{equation}
\label{def:hat-Gamma}
  \widehat{\Gamma}(a):=\mathcal B(a)\mathcal B(a)^*
\end{equation}
$\mathcal B(a)$ is the left-trivialization of the derivative $DF_E(a)$, 
and $\widehat{\Gamma}(a)$ is the left-trivialization of the Malliavin 
covariance given by 
\begin{equation}
\label{def:Gamma}
  \Gamma(a):=DF_E(a)DF_E(a)^*. 
\end{equation}
For the proof of Proposition \ref{prp:Fisher-bound} 
we will need estimates on the following quantities. 
For any $a\in \mathcal H_W$ define
\begin{equation}
\label{def:Lambda-left-trivialized}
\begin{aligned}
  \Lambda_0(a)&:=\|\widehat{\Gamma}(a)^{-1}\|_{\mathfrak{sp}_{2W}\to
  \mathfrak{sp}_{2W}},\\
  \Lambda_1(a)&:= \|\mathcal B(a)\|_{\mathcal H_W\to \mathfrak{sp}_{2W}},\\
  \Lambda_2(a)&:=
  \sup_{\|N\|_{\mathfrak{sp}_{2W}}=1}
  \|D(\mathcal B^*)(a)N\|_{HS(\mathcal H_W,\mathcal H_W)}\\
  &\quad+
  \sup_{\|h\|_{\mathcal H_W}=1}
  \|D\mathcal B(a)[\,\cdot\,]h\|_{HS(\mathcal H_W,\mathfrak{sp}_{2W})}.
\end{aligned}
\end{equation}
We briefly explain the domains and codomains in the norms of
\eqref{def:Lambda-left-trivialized}. Since
$\mathcal B:\mathcal H_W\to
\mathcal L(\mathcal H_W,\mathfrak{sp}_{2W}),$ 
its derivative at \(a\) is a linear map
\[
  D\mathcal B(a):
  \mathcal H_W\to
  \mathcal L(\mathcal H_W,\mathfrak{sp}_{2W}).
\]
Here we have used the standard identification of a vector 
space with its tangent space. Thus, for \(k\in\mathcal H_W\), \(D\mathcal B(a)[k]\) is a
linear operator from \(\mathcal H_W\) to \(\mathfrak{sp}_{2W}\). For fixed
\(h\in\mathcal H_W\), we write \(D\mathcal B(a)[\,\cdot\,]h\) for the
linear map
\[
  k\longmapsto D\mathcal B(a)[k]h.
\]
Similarly, the derivative of $\mathcal B^{\ast}$ at $a$
is the linear mapping
$$D(\mathcal B^*)(a):
  \mathcal H_W\to
  \mathcal L(\mathfrak{sp}_{2W},\mathcal H_W).$$
Thus for fixed \(N\in\mathfrak{sp}_{2W}\), we can regard
$D(\mathcal B^*)(a)N:\mathcal H_W\to \mathcal H_W$ 
as the linear operator
\[
  h\longmapsto D(\mathcal B^*)(a)[h]N. 
\]
Although this is a bit abstract, 
the claims below give concrete matrix formulas for all 
these objects, and those formulas are what we use 
to analyze them. The proposition below estimates these
norms, together with the operator norms of
$\mathcal{B}(a)$ and $\widehat{\Gamma}(a)^{-1}$.

\begin{proposition}
\label{prp:operator-estimates}
Let $E\in\mathbb R$, $W\in\mathbb N$, and let $a\in\mathcal H_W$ have the
law of $1000$ independent normalized GOE matrices. For each $p>0$ there
exists $C:=C(E,p)$
such that for all $W$ large enough, depending on $p$,
\begin{equation}
\label{ineq:Lambda-moments}
  \mathbb E_a\left[
    \Lambda_0(a)^p+\Lambda_1(a)^p+W^{-p/2}\Lambda_2(a)^p
  \right]\leq C. 
\end{equation}
\end{proposition}

\begin{proof}
To estimate $\Lambda_1$ and $\Lambda_2$
we compute the action of these operators. 
Straightforward, but tedious, 
computations give that these 
are given by sums and products 
involving $T_E(A)$ and its inverse, 
and so one can estimate them by applying 
standard operator 
norm estimates. $\Lambda_0$ is 
more difficult. We estimate 
it by computing 
the quadratic form of $\widehat{\Gamma}(a)$, 
then lower bound that in terms 
of the least singular value of a random operator.
The needed estimate on the least singular value 
is proven in the appendix.  

We use the standard normalized GOE operator-norm estimate: for every $p>0$
there exists $C:=C(p)>0$, independent of $W$, such that
\begin{equation}
\label{ineq:GOE-operator-norm-estimate}
\mathbb{E}\|A\|_{\mathbb{R}^{W}\to\mathbb{R}^{W}}^p\leq C.
\end{equation}
Hence, by Cauchy--Schwarz, 
it suffices to prove the following deterministic bounds:
\[
\begin{aligned}
\Lambda_1(a)&\leq
C\left(1+|E|+\sup_{i=1,\ldots,1000}\|a_i\|_{op}\right)^C,\\
\Lambda_2(a)&\leq
C\sqrt W\left(1+|E|+\sup_{i=1,\ldots,1000}\|a_i\|_{op}\right)^C. 
\end{aligned}
\]

The action of $\mathcal B(a)$ is computed in the following claim. 
\begin{claim}
For any \(a=(a_1,\ldots,a_{1000})\), \(h=(h_1,\ldots,h_{1000})\in \mathcal H_W\), we have
\begin{equation}
\label{eq:action-of-B}
\mathcal B(a)[h]
= \sum_{i=1}^{1000} \operatorname{Ad}_{Q_i(a)^{-1}} \mathfrak n(h_i),
\end{equation}
where
\[
  Q_i(a):=T_E(a_{i+1})\cdots T_E(a_{1000}),
\]
with \(Q_{1000}(a)=\Id\), and
\begin{equation}
\label{def:n}
\mathfrak n(H):=
  \begin{pmatrix}
    0&0\\ -H&0
  \end{pmatrix}.
\end{equation}
\end{claim}
\begin{proof}
By calculus we have
\begin{align*}
DF_E(a)[h] = \frac{d}{dt}\Big|_{t=0}F_E(a+th) 
 = \sum_{i=1}^{1000}
T_E(a_1)\cdots T_E(a_{i})\mathfrak n(h_i)T_E(a_{i+1})\cdots T_E(a_{1000}).
\end{align*}
Since \(d(L_{F_E(a)^{-1}})_{F_E(a)}\) acts by left multiplication by
\(F_E(a)^{-1}\), we have
\begin{align*}
\mathcal B(a)[h]
&=F_E(a)^{-1}(DF_E(a)[h])\\
&=(T_E(a_1)\cdots T_E(a_{1000}))^{-1}
\sum_{i=1}^{1000}
T_E(a_1)\cdots T_E(a_i)\mathfrak n(h_i)T_E(a_{i+1})\cdots T_E(a_{1000})\\
&=\sum_{i=1}^{1000}
\operatorname{Ad}_{(T_E(a_{i+1})\cdots T_E(a_{1000}))^{-1}}\mathfrak n(h_i).
\end{align*}
\end{proof}
The estimate on \(\Lambda_1(a)\) now follows from
$$
  T_E(A)^{-1}
  =
  \begin{pmatrix}
    0&\Id_W\\
    -\Id_W&A-E\Id_W
  \end{pmatrix},
$$
since
\begin{equation}
\|T_E(A)^{-1}\|_{\mathbb{R}^{2W}\to \mathbb{R}^{2W}}, \|T_E(A)\|_{\mathbb{R}^{2W}\to \mathbb{R}^{2W}}
\leq C\left(1+|E|+\|A\|_{op}\right),
\end{equation} 
so 
\begin{equation}
\label{ineq:elementary-transfer-matrix-estimate}
\|Q_i(a)\|_{\mathbb{R}^{2W}\to \mathbb{R}^{2W}}, \|Q_i(a)^{-1}\|_{\mathbb{R}^{2W}\to \mathbb{R}^{2W}}
\leq C\left(1+|E|+\sup_{\ell=1,\ldots,1000}\|a_\ell\|_{op}\right)^C.
\end{equation}

The following computes the action of \(D\mathcal B(a)\) and \(D(\mathcal B^*)(a)\).
To define it we introduce the operator
$\pi_-:\mathfrak{sp}_{2W}\to\mathbb R_{\mathrm{Sym}}^{W\times W}$ 
given by 
\begin{equation}
\label{def:pi-}
\pi_-\left(\begin{pmatrix}
    P&Q\\
    R&-P^T
  \end{pmatrix}\right):=R.
\end{equation}
Here we have used the representation of the Lie algebra $\mathfrak{sp}_{2W}$
in \eqref{def:Lie-Algebra-representation}. 
In particular, we know $\pi_-$ takes values in the symmetric matrices. 
\begin{claim}
For any \(a,h,k\in \mathcal H_W\), we have
\[
  D\mathcal B(a)[k]h
  =
  \sum_{i=1}^{1000}
  \left[
    \operatorname{Ad}_{Q_i(a)^{-1}}\mathfrak n(h_i),
    \sum_{j=i+1}^{1000}
    \operatorname{Ad}_{Q_j(a)^{-1}}\mathfrak n(k_j)
  \right],
\]
where $[X,Y]= XY-YX,$ denotes the Lie bracket. Moreover, for \(Y\in\mathfrak{sp}_{2W}\),
\[
\begin{aligned}
  D(\mathcal B^*)(a)[k]Y
  &=
  \Bigl(
    -\pi_-\left(
      \operatorname{Ad}_{Q_i(a)^{-T}}
      \left[
      Y,
      \left(
        \sum_{j=i+1}^{1000}
        \operatorname{Ad}_{Q_j(a)^{-1}}\mathfrak n(k_j)
      \right)^T
      \right]
    \right)
  \Bigr)_{i=1}^{1000}. 
\end{aligned}
\]
Recall that $\mathfrak n$ is defined in \eqref{def:n}. 
\end{claim}
\begin{proof}
\(D\mathcal B(a)[k]h\) captures 
how $\mathcal B(a)[h]\in \mathfrak{sp}_{2W}$ changes
as we move $a$ in the direction $k\in \mathcal H_W$, 
i.e. it is given by
\[
  D\mathcal B(a)[k]h
  =
  \left.\frac{d}{dt}\right|_{t=0}\mathcal B(a+tk)[h].
\]
Hence, we compute 
$D\mathcal{B}(a)[k]h$ by differentiating 
each term in sum 
\eqref{eq:action-of-B} with respect 
to $a$ in the direction $k$.

For the $i$-th term \eqref{eq:action-of-B} 
a short computation gives
\begin{align*}
\frac{d}{dt}\Big|_{t=0}\left(Q_i(a+tk)^{-1}\mathfrak n(h_i)Q_i(a+tk)\right)
& = \left[
    \operatorname{Ad}_{Q_i(a)^{-1}}\mathfrak n(h_i),
    Q_i(a)^{-1}DQ_i(a)[k]
  \right].
\end{align*}
Since, by the proof of the previous claim,
\[
  Q_i(a)^{-1}DQ_i(a)[k]
  =
  \sum_{j=i+1}^{1000}
  \operatorname{Ad}_{Q_j(a)^{-1}}\mathfrak n(k_j),
\]
the expression for $D\mathcal B$ follows. 

\(D(\mathcal B^*)\) can be computed similarly 
by taking the adjoint of \(\mathcal B(a)\), given in
\eqref{eq:action-of-B}, to get
\begin{equation}
\label{eq:B*-equation}
\mathcal B(a)^*M
  =
  \left(
    -\pi_-(\operatorname{Ad}_{Q_i(a)^{-T}}M)
  \right)_{i=1}^{1000}.
\end{equation}
Then direct differentiation on each component gives
\[
\left.\frac{d}{dt}\right|_{t=0}
\operatorname{Ad}_{Q_i(a+tk)^{-T}}M
=
\operatorname{Ad}_{Q_i(a)^{-T}}
\left[
M,\left(Q_i(a)^{-1}DQ_i(a)[k]\right)^T
\right].
\]
Substituting the expression for \(Q_i(a)^{-1}DQ_i(a)[k]\) from above proves
the formula for \(D(\mathcal B^*)(a)[k]M\).
\end{proof}

To estimate the Hilbert--Schmidt norm of the mapping
\(k\mapsto D\mathcal B(a)[k]h\), note that both sums in Claim 4.7 have a
uniformly bounded number of terms. Thus, by the triangle inequality, it
suffices to estimate, for any fixed
\(1\leq i<j\leq1000\),
\[
\left\|
k\longmapsto
\left[
  \operatorname{Ad}_{Q_i(a)^{-1}}\mathfrak n(h_i),
  \operatorname{Ad}_{Q_j(a)^{-1}}\mathfrak n(k_j)
\right]
\right\|_{HS(\mathcal H_W,\mathfrak{sp}_{2W})}.
\]
At this point we note that for any \(d\in\mathbb N\) and fixed matrices
\(L,R\in\mathbb R^{d\times d}\), we have
\[
\begin{aligned}
  \|H\mapsto LHR\|_{HS(\mathbb R_{\mathrm{Sym}}^{d\times d},
  \mathbb R^{d\times d})}^2
  & \leq \|H\mapsto LHR\|_{HS(\mathbb R^{d\times d},
  \mathbb R^{d\times d})}^2\\
  & = \sum_{1\leq p,q\leq d} \|L(e_pe_q^{\ast}) R\|_{HS}^2\\
  & = \sum_{1\leq p,q\leq d} \|Le_p\|^2\|R^{*}e_q\|^2\\
  &\leq 
  \|L\|_{HS(\mathbb{R}^{d\times d})}^2\|R\|_{HS(\mathbb{R}^{d\times d})}^2\\
  & \leq
  d\|L\|_{HS(\mathbb{R}^{d\times d})}^2\|R\|_{op}^2.
\end{aligned}
\]
where $\|Le_p\|$ and $\|R^*e_q\|$ denote standard $\ell^2$ norms,
and $\|R\|_{op}$ is the operator norm of $R$ as a mapping
$\mathbb R^d\to\mathbb R^d$. We apply this estimate below with $d=2W$.
For fixed \(i<j\), expanding the commutator and the two adjoint actions and
applying the preceding estimate (or its transposed version), together with
\eqref{ineq:elementary-transfer-matrix-estimate}, and using 
that $\|h_i\|_{HS}\leq 1$ gives
\begin{align*}
&\left\|
k\longmapsto\left[
  \operatorname{Ad}_{Q_i(a)^{-1}}\mathfrak n(h_i),
\operatorname{Ad}_{Q_j(a)^{-1}}\mathfrak n(k_j)
\right]
\right\|_{HS(\mathcal H_W,\mathfrak{sp}_{2W})}^2\\
&\qquad\leq
  C W\left(1+|E|+\sup_{\ell=1,\ldots,1000}\|a_\ell\|_{op}\right)^C.
\end{align*}
which implies the deterministic estimate for $\|D\mathcal B (a)[\cdot] h\|_{HS(\mathcal H_W, \mathfrak sp_{2W})}.$

The displayed formula for \(D(\mathcal B^*)\) is handled in the same way, using
that \(\pi_-\) has norm at most one and transposition preserves the
Hilbert--Schmidt norm. This proves the deterministic bound on \(\Lambda_2(a)\).

It remains to estimate 
$$\Lambda_0(a) := \|\widehat{\Gamma}(a)^{-1}\|_{\mathfrak{sp}_{2W}\to \mathfrak{sp}_{2W}}.$$ 
This 
is equivalent to estimating the smallest
eigenvalue of the positive semi-definite matrix \(\widehat\Gamma(a)\), 
and we approach this via the quadratic form of $\widehat{\Gamma}(a)$ 
computed below. 

\begin{claim}
For any \(M\in \mathfrak{sp}_{2W}\), we have
\begin{equation}
\label{eq:quadratic-form-comp}
\left\langle \widehat{\Gamma}(a)M,M\right\rangle_{\mathfrak{sp}_{2W}}
  =
  \sum_{i=1}^{1000}
  \left\|
    \pi_-(\operatorname{Ad}_{Q_i(a)^{-T}}M)
  \right\|_{HS}^2 .
\end{equation}
Recall $\pi_{-}(N)$, defined in \eqref{def:pi-},
is the restriction
to the lower left block of $N$. 
\end{claim}
\begin{proof}
For any \(M\in \mathfrak{sp}_{2W}\), we have 
\[
\left\langle\widehat{\Gamma}(a)M,M\right\rangle_{\mathfrak{sp}_{2W}}
=
\left\langle \mathcal B(a)^*M,\mathcal B(a)^*M\right\rangle_{\mathcal H_W}.
\]
Recalling the action of $\mathcal{B}^{\ast}$, 
as given in \eqref{eq:B*-equation}, gives the result. 
\end{proof}

We want to lower bound the RHS of the above 
in terms of $\|M\|_{HS}^2$. 
To do this we exploit the recursive relations between 
the matrices on the RHS. 

Fix any \(M\in \mathfrak{sp}_{2W}\) 
and use the representation of
\(\mathfrak{sp}_{2W}\) from \eqref{def:Lie-Algebra-representation}, 
to write
\[
M=
\begin{pmatrix}
P&U\\
R&-P^T
\end{pmatrix}
\]
and
\[
\operatorname{Ad}_{Q_i(a)^{-T}}M
=
\begin{pmatrix}
P_i&U_i\\
R_i&-P_i^T
\end{pmatrix},
\]
for some \(U,R,U_i,R_i\in \mathbb R_{\mathrm{Sym}}^{W\times W}\) and $P, P_i\in \mathbb{R}^{W\times W}$.
The matrices \(U_i,P_i,R_i\)
satisfy a recursive relation since
\[
  \operatorname{Ad}_{Q_{i-1}(a)^{-T}}M
  =
  \operatorname{Ad}_{T_E(a_i)^{-T}}
  \left(\operatorname{Ad}_{Q_i(a)^{-T}}M\right),
\]
so matrix multiplication gives
for \(i=2,\ldots,1000\),
\[
\begin{aligned}
  P_{i-1}&=-R_iC_i-P_i^T,\\
  U_{i-1}&=-R_i,\\
  R_{i-1}&=P_iC_i+C_iP_i^T+C_iR_iC_i-U_i,
\end{aligned}
\]
where \(C_i=a_i-E\Id_W\), and we started from
\(P_{1000}=P\), \(U_{1000}=U\), and \(R_{1000}=R\). 

The RHS of \eqref{eq:quadratic-form-comp} is
$\sum_{i=1}^{1000}\|R_i\|_{HS}^2,$
so we would like to relate the $R_i$ to $R,U,P$. 
To do this iterate the 
first recursive relation to get 
\[
  P_i
  =
  \mathsf S^{1000-i}(P)
  -
  \sum_{j=i+1}^{1000}
  \mathsf S^{j-i-1}\left(R_jC_j\right),
\]
where we used the operator $\mathsf S$ given by
$$\mathsf S(X)=-X^T.$$
Since $\mathsf S$ preserves HS-norm and the number of iterates is fixed,
\[
  \left\|P_i-\mathsf S^{1000-i}(P)\right\|_{HS}
  \leq
  C\left(1+\sup_j\|C_j\|_{op}\right)
  \left(\sum_{j=1}^{1000}\|R_j\|_{HS}^2\right)^{1/2}.
\]
Hence for \(i=487,489,\ldots,997\), 
the second recursion with index \(i+1\)
and the third recursion with index \(i\), together with
$\mathsf S^{1000-i}(P)=-P^T$, gives

\begin{align*}
  \|C_iP+P^TC_i\|_{HS}
  &\leq \|P_iC_i+C_iP_i^T\|_{HS}
  +2\|C_i\|_{op}\|P_i-\mathsf S^{1000-i}(P)\|_{HS}\\
  &=\|R_{i-1}-C_iR_iC_i-R_{i+1}\|_{HS}
  +2\|C_i\|_{op}\|P_i-\mathsf S^{1000-i}(P)\|_{HS}\\
  & \leq C\left(1+\sup_j\|C_j\|_{op}\right)^C\sum_{j=1}^{1000}\|R_j\|_{HS}.
\end{align*}

For $1\leq j\leq128$, put
$$
G_j:=\frac{C_{4j+485}-C_{4j+483}}{\sqrt2}.
$$
Define the linear operator
\(\mathcal L_a:\mathbb R^{W\times W}\to
(\mathbb R_{\mathrm{Sym}}^{W\times W})^{128}\) by
\begin{equation}
\label{def:L-block-anderson}
\mathcal L_a(P)
:=
\left(G_jP+P^TG_j\right)_{j=1}^{128}.
\end{equation}
The preceding estimate and the triangle inequality give
\begin{align*}
\|P\|_{HS} 
& \leq s_{min}(\mathcal L_a)^{-1}\|\mathcal L_a(P)\|_{HS}\\
& \leq Cs_{min}(\mathcal L_a)^{-1} \left(1+\sup_j\|C_j\|_{op}\right)^C\sum_{j=1}^{1000}\|R_j\|_{HS}.
\end{align*}
The first line follows from the definition of the least singular value,
and the second follows from the preceding estimate and the definition of
$\mathcal L_a$.
For $R$ and $U$ we note that $R_{1000}= R$ so 
$$\|R\|_{HS} \leq \sum_{j=1}^{1000}\|R_j\|_{HS},$$
and the third recursion relation with $i=1000$ gives
$$
  U=P C_{1000}+C_{1000}P^T+C_{1000}RC_{1000}-R_{999}.
$$
so 
\begin{align*}
\|U\|_{HS}
&\leq C(1+s_{min}(\mathcal L_a)^{-1})
\left(1+\sup_j\|C_j\|_{op}\right)^C
\sum_{j=1}^{1000}\|R_j\|_{HS}.
\end{align*}

The above estimate on $U,P,R$ and 
\eqref{eq:quadratic-form-comp} gives
that for any $M$ we have 
\begin{align*}
\|M\|_{HS}^2 
& = 2\|P\|_{HS}^2  + \|U\|_{HS}^2 + \|R\|_{HS}^2\\
& \leq C(1+s_{min}(\mathcal L_a)^{-2})\left(1+\sup_j\|C_j\|_{op}\right)^C\sum_{j=1}^{1000}\|R_j\|_{HS}^2\\
& \leq C(1+s_{min}(\mathcal L_a)^{-2})\left(1+\sup_j\|C_j\|_{op}\right)^C \langle \widehat\Gamma(a)M,M\rangle_{\mathfrak{sp}_{2W}}.
\end{align*}
Thus 
$$
\Lambda_0(a)
\leq
C(1+s_{min}(\mathcal L_a)^{-2})
\left(1+\sup_j\|C_j\|_{op}\right)^C.
$$

To finish the bound on $\Lambda_0(a)$, 
we note that Proposition 
\ref{prp:structured-singular-value-estimate}
implies that for any $p>0$, there exists $C:=C(E,p)>0$ such that
\begin{equation}
\label{ineq:final-singular-value-estimate}
\mathbb{E}s_{min}(\mathcal L_a)^{-p}\leq C,
\end{equation}
when $W$ is sufficiently large.
Indeed, $C_i=a_i-E\Id_W$, and
the matrices $G_1,\ldots,G_{128}$ are independent normalized GOE matrices.
Thus Proposition \ref{prp:structured-singular-value-estimate} gives
\eqref{ineq:final-singular-value-estimate}. This completes the proof.
\end{proof}

\subsection{Proof of Proposition \ref{prp:Fisher-bound}}

\begin{proof}[Proof of Proposition \ref{prp:Fisher-bound}]
Fix \(M\in\mathfrak{sp}_{2W}\) with 
\(\|M\|_{HS}=1\), 
and define the vector field 
\(\tilde X_M:\mathcal H_W \to \mathcal H_W\) 
by
\begin{equation}
\label{def:X-tilde}
\tilde{X}_M(a):= DF_E(a)^{\ast}\left(DF_E(a)DF_E(a)^{\ast}\right)^{-1}X_M(F_E(a)). 
\end{equation} 
This is equivalent to 
$$\tilde X_M(a) = \mathcal B(a)^*\widehat{\Gamma}(a)^{-1}M,$$
where
$\widehat{\Gamma}(a)^{-1}$ and $\mathcal B(a)$ are defined 
in
\eqref{def:hat-Gamma} and \eqref{def:B-of-a}. 

Proposition \ref{prp:operator-estimates} implies
$\tilde{X}_M$ exists almost everywhere, and the estimates there 
also imply  
$\tilde{X}_M\in W^{1,1}_{loc}(\mathcal H_W; \mathcal H_W),$
via a standard regularization argument. 
To check that $\tilde{X}_M$ is a \textit{lift} of $M$ in the sense of \eqref{def:lifting-relation}, 
note that for any $a \in \mathcal H_W$ where $\widehat{\Gamma}(a)$ is invertible, 
we have  
\begin{align*}
DF_E(a)\tilde X_M(a)
& =DF_E(a)\mathcal B(a)^*\widehat{\Gamma}(a)^{-1}M\\
& =
d(L_{F_E(a)})_{\Id}
\mathcal B(a)\mathcal B(a)^*\widehat{\Gamma}(a)^{-1}M\\
& =d(L_{F_E(a)})_{\Id} M\\
& = X_M(F_E(a)).
\end{align*}
We claim it suffices to prove the estimate 
\begin{equation}
\label{eq:vf-estimate-goal}
W\mathbb E_a\|\tilde X_M(a)\|_{\mathcal H_W}^2
  +
  \mathbb E_a\|D\tilde X_M(a)\|_{HS(\mathcal H_W)}^2\leq CW,
\end{equation}
for some $C$, independent of $W$ and $M$. 
Indeed, the covariance of $a$ as a Gaussian vector in $\mathcal H_W$ is
$2W^{-1}\Id$. Thus Lemma \ref{lem:abstract-gaussian-lifting} gives
$$
  \|X_M\mu^{\ast1000}\|_{TV}^2
  \leq
  \frac W2\mathbb E_a\|\tilde X_M(a)\|_{\mathcal H_W}^2
  +\mathbb E_a\|D\tilde X_M(a)\|_{HS(\mathcal H_W)}^2,
$$
and so \eqref{eq:vf-estimate-goal} would give
\[
  \|X_M\mu^{\ast 1000}\|_{TV(\Sp_{2W}(\mathbb R))}^2
  \leq CW,
\]
for any $M\in \mathfrak{sp}_{2W}$ with $\|M\|_{HS}\leq 1$.  
This implies that 
$\mu^{\ast 1000}$ has a density with respect to Haar measure, 
and so we can replace the $TV$-norm 
quantity above with $\|X_M\frac{d\mu^{\ast 1000}}{dm}\|_{L^1(\Sp_{2W}(\mathbb{R}))}.$

It remains to prove \eqref{eq:vf-estimate-goal}. To do this we
compute and use Proposition \ref{prp:operator-estimates}. 
We use the same notation as Proposition \ref{prp:operator-estimates}, 
namely we define

$$
\begin{aligned}
  \Lambda_0(a)&:=\|\widehat{\Gamma}(a)^{-1}\|_{\mathfrak{sp}_{2W}\to
  \mathfrak{sp}_{2W}},\\
  \Lambda_1(a)&:= \|\mathcal B(a)\|_{\mathcal H_W\to \mathfrak{sp}_{2W}},\\
  \Lambda_2(a)&:=
  \sup_{\|N\|_{\mathfrak{sp}_{2W}}=1}
  \|D(\mathcal B^*)(a)N\|_{HS(\mathcal H_W,\mathcal H_W)}\\
  &\quad+
  \sup_{\|V\|_{\mathcal H_W}=1}
  \|D\mathcal B(a)[\,\cdot\,]V\|_{HS(\mathcal H_W,\mathfrak{sp}_{2W})}.
\end{aligned}
$$
In terms of the above quantities, we have the following estimate. 
\begin{claim}
For almost every $a\in \mathcal H_W$
$$\|\tilde X_M(a)\|_{\mathcal H_W}^2 \leq \Lambda_0(a).$$
\end{claim}
\begin{proof}
$\tilde{X}_M(a)$ is defined for almost every $a$, 
and for all such $a$ we have
\begin{align*}
\|\tilde X_M(a)\|_{\mathcal H_W}^2
& = \langle \mathcal B(a)^*\widehat{\Gamma}(a)^{-1}M, \mathcal B(a)^*\widehat{\Gamma}(a)^{-1}M\rangle_{\mathcal H_W}\\
&=\left\langle\widehat{\Gamma}(a)^{-1}M,M\right\rangle_{\mathfrak{sp}_{2W}}\\
&\leq \Lambda_0(a)\|M\|_{HS}^2\\ 
& \leq \Lambda_0(a).
\end{align*}
\end{proof}

We can estimate the Hilbert-Schmidt 
norm of $D\tilde{X}_M$ as follows.  
\begin{claim}
For almost every $a\in \mathcal H_W$
$$
  \|D\tilde{X}_M(a)\|_{HS}
  \leq 2\Lambda_2(a)\Lambda_0(a).
$$
\end{claim}
\begin{proof}
Differentiating, we compute
\[
  D(\widehat{\Gamma}^{-1})(a)
  =
  -\widehat{\Gamma}(a)^{-1}
  D\widehat{\Gamma}(a)
  \widehat{\Gamma}(a)^{-1}
\]
and
\[
  D\widehat{\Gamma}(a)
  =
  D\mathcal B(a)\mathcal B(a)^*
  +
  \mathcal B(a)D(\mathcal B^*)(a).
\]
Thus differentiating
\(\tilde X_M(a)=\mathcal B(a)^*\widehat{\Gamma}(a)^{-1}M\) gives
\[
\begin{aligned}
  D\tilde X_M(a)
  & = D(\mathcal B(a)^{\ast})\widehat{\Gamma}(a)^{-1}M + \mathcal B(a)^{\ast}D(\widehat{\Gamma}(a)^{-1})M\\
  &=
  \left(
    I-\mathcal B(a)^*\widehat{\Gamma}(a)^{-1}\mathcal B(a)
  \right)
  D(\mathcal B^*)(a)\widehat{\Gamma}(a)^{-1}M\\
  &\quad-
  \mathcal B(a)^*\widehat{\Gamma}(a)^{-1}
  D\mathcal B(a)[\,\cdot\,]\tilde X_M(a).
\end{aligned}
\]
For the first term we note that 
\[
  \mathcal B(a)^*\widehat{\Gamma}(a)^{-1}\mathcal B(a)
\]
is the orthogonal projection onto \(\operatorname{Range}(\mathcal B(a)^*)\).
Hence $I-\mathcal B(a)^*\widehat{\Gamma}(a)^{-1}\mathcal B(a)$, as a mapping $\mathcal H_W\to \mathcal H_W$, 
has operator norm at most one. Using 
Proposition
\ref{prp:operator-estimates} then gives
\begin{align*}
&\left\|\left(
    I-\mathcal B(a)^*\widehat{\Gamma}(a)^{-1}\mathcal B(a)
  \right)
  D(\mathcal B^*)(a)\widehat{\Gamma}(a)^{-1}M\right\|_{HS(\mathcal H_W)}\\
&\leq
\left\|
    D(\mathcal B^*)(a)\widehat{\Gamma}(a)^{-1}M
  \right\|_{HS(\mathcal H_W)}\\
&\leq
\Lambda_2(a)\Lambda_0(a).
\end{align*}
For the second term, first note that
\begin{align*}
  \|\mathcal B(a)^*\widehat{\Gamma}(a)^{-1}\|_{op}^2
  &=
  \|\widehat{\Gamma}(a)^{-1}\mathcal B(a)\mathcal B(a)^*
  \widehat{\Gamma}(a)^{-1}\|_{op}\\
  &=\|\widehat{\Gamma}(a)^{-1}\|_{op}
  =\Lambda_0(a).
\end{align*}
Consequently,
\begin{align*}
  \left\|
    \mathcal B(a)^*\widehat{\Gamma}(a)^{-1}
    D\mathcal B(a)[\,\cdot\,]\tilde X_M(a)
  \right\|_{HS}
  &\leq
  \Lambda_0(a)^{1/2}
  \|D\mathcal B(a)[\,\cdot\,]\tilde X_M(a)\|_{HS}\\
  &\leq
  \Lambda_0(a)^{1/2}\Lambda_2(a)
  \|\tilde X_M(a)\|_{\mathcal H_W}\\
  &\leq \Lambda_2(a)\Lambda_0(a).
\end{align*}
Combining the two estimates proves the claim.
\end{proof}
The estimates in Proposition \ref{prp:operator-estimates} 
combined with the above claims, imply 
$$W\mathbb{E}_a\|\tilde{X}_M(a)\|_{\mathcal H_W}^2 + \mathbb{E}_a\|D\tilde{X}_M(a)\|_{HS(\mathcal H_W)}^2 \leq C W.$$
\end{proof}

\section{Proof of Theorem \texorpdfstring{\ref{thm:gaps}}{Gap Separation} and Corollary \ref{cor:RBM-Localization}}
\begin{proof}[Proof of Theorem \ref{thm:gaps}]
We assume $W$ is sufficiently large, 
as the theorem follows 
for any bounded $W$ by the 
Margulis-Goldshield criterion \cite{Margulis87}, 
decreasing $c$ as necessary. 

Let \(\mu\) be the law of the random matrix in
\eqref{def:M} and $L_1(E,W)\geq ... \geq L_{2W}(E,W)$ be the associated Lyapunov exponents. By Proposition
\ref{prp:Fisher-bound}, $\mu^{\ast 1000}$ has a density with respect 
to Haar measure $dm$, and 
\[
  \sup_{\substack{M\in \mathfrak{sp}_{2W}\\ \|M\|_{HS}=1}}
  \|X_M \frac{d\mu^{\ast 1000}}{dm}\|_{L^1(\Sp_{2W}(\mathbb R))}^2\leq CW.
\]
Applying Theorem \ref{thm:Fisher-gaps} to \(\mu^{*1000}\), 
and using that the Lyapunov exponents of
\(\mu^{*1000}\) are \(1000L_k(E,W)\), gives 
\[
  1000\bigl(L_k(E,W)-L_{k+1}(E,W)\bigr)\geq \frac{c}{W}
\]
for every \(1\leq k\leq 2W-1\), which implies the result.
\end{proof}

The following is classical, and essentially follows \cite{KleinLacroixSpeis1990}. 

\begin{proof}[Proof of Corollary \ref{cor:RBM-Localization}]
Let $A_1, A_2, \ldots$ and $H_{W,n}$ be as in Corollary
\ref{cor:RBM-Localization}, and $T_E(A)$ be 
as in \eqref{def:M}. First we recall that for any $E\in \mathbb{R}$ we have 
$$\lim_{n\to\infty}\frac{1}{n}\log \|(H_{W,n}-E)^{-1}(1,n)\|_{op} = -L_{W}(E,W),$$
almost surely. Note $(H_{W,n}-E)^{-1}$ exists for all $n$ almost surely 
by the Wegner estimate, see Proposition 2 in \cite{DroginRBMLocalization2025}, 
for instance. 

Indeed, if 
$$Q_n := (I,0)T_E(A_n)\cdots T_E(A_1)
\begin{pmatrix}
I\\
0
\end{pmatrix},$$
then page 143 of \cite{KleinLacroixSpeis1990} 
gives
$$
(H_{W,n}-E)^{-1}(1,n)
=
(-1)^{n-1}
Q_n^{-1}.
$$
Hence,
\begin{align*}
\lim_{n\to\infty}\frac{1}{n}\log \left\|(H_{W,n}-E)^{-1}(1,n)\right\|_{\mathrm{op}}
& =\lim_{n\to \infty}\frac{1}{n}\log \left\|Q_n^{-1}\right\|_{\mathrm{op}}\\
& =-\lim_{n\to \infty}\frac{1}{n}\log\sigma_W(Q_n).
\end{align*}

To estimate $\sigma_{W}(Q_n)$ we invoke Theorem \ref{thm:gaps} 
and classical results from Lyapunov theory. 
Proposition 2.2 of Chapter 6 
of \cite{BougerolLacroix1985} 
gives that if  
$M_1,M_2,...\in \mathbb{R}^{d\times d}$ 
are iid random matrices
having a probability distribution with
irreducible and proximal support then 
for any fixed unit vectors $v,w\in \mathbb{R}^{d}$
$$\lim_{n\to \infty} \frac{1}{n}\log |\langle v, M_n...M_1 w\rangle| = \lambda_1,$$
where the limit is meant almost surely and $\lambda_1$ is the top Lyapunov exponent associated to $M_1$. 
Since $T_E(A_1)...T_E(A_{1000})$ has a density in $\Sp_{2W}(\mathbb{R})$, 
we can apply this to the $W$-th and $(W-1)$-th exterior products of $T_E(A_n)...T_E(A_1)$ 
to get
$$\lim_{n\to\infty}\frac{1}{n}\log \sigma_{W}(Q_n) = \lim_{n\to\infty} \frac{1}{n}\log \frac{|\operatorname{det} Q_n|}{\|\wedge^{W-1} W_n\|} = L_W(E,W).$$

Now the Corollary follows by Theorem \ref{thm:gaps} 
since $L_{W}(E,W)-L_{W+1}(E,W) = 2L_{W}(E,W),$ 
as $T_E(A)$ is Symplectic. 
\end{proof}

\appendix

\section{Singular Value Estimate}

For any \(A_1,\ldots,A_m\in\mathbb R_{\mathrm{sym}}^{W\times W}\), 
define the linear mapping 
$$\mathcal L_{A_1,\ldots,A_m}:\mathbb{R}^{W\times W}\to (\mathbb{R}^{W\times W}_{sym})^{m}$$ 
by 
\[
  \mathcal L_{A_1,\ldots,A_m}(P)
  :=
  \left(A_iP+P^TA_i\right)_{i=1}^m .
\]
We prove the following least singular value estimate. 
\begin{proposition}\label{prp:structured-singular-value-estimate}
For every $p>0$ there exist $C:=C(p)>0$ and
$W_0:=W_0(p)\in\mathbb N$ such that the following holds.
Suppose $W\geq W_0$ and $A_1,\ldots,A_{128}$ are independent
normalized GOE matrices of size $W$. Then
\[
  \mathbb E\left[
  s_{min}(\mathcal L_{A_1,\ldots,A_{128}})^{-p}
  \right]
  \leq C.
\]
\end{proposition}

The proof uses two lemmas. The first shows that
$s_{min}(\mathcal L_{A_1,\ldots,A_{128}})\geq c$ with high
probability.
\begin{lemma}
\label{lem:fixed-scale} 
There exist $c,\tau>0$ such that
$$
\mathbb P\left(
s_{min}(\mathcal L_{A_1,\ldots,A_{128}})\leq\tau
\right)
\leq e^{-cW}
$$
for all sufficiently large $W$.
\end{lemma}

The second estimates the lower tail of $s_{min}$ near $0$.
\begin{lemma}
\label{lem:small-t-estimate}
There exist $c,C>0$ and $W_0\in\mathbb N$ such that for any
$W\geq W_0$ and $0<t<e^{-C\sqrt{W}}$ we have 
$$\mathbb{P}(s_{min}(\mathcal L_{A_1,...,A_{128}})\leq t)\leq t^{c\sqrt{W}}.$$
\end{lemma}

Now we prove Proposition \ref{prp:structured-singular-value-estimate}
assuming the above two lemmas, which are proved afterward.
\begin{proof}[Proof of Proposition \ref{prp:structured-singular-value-estimate}]
By Lemmas \ref{lem:fixed-scale} and \ref{lem:small-t-estimate}, after
adjusting the absolute constants, for every $0<t<\tau$ we have
$$
\mathbb P\left(
s_{min}(\mathcal L_{A_1,\ldots,A_{128}})\leq t
\right)
\leq t^{c\sqrt W}.
$$
Indeed, this follows from Lemma \ref{lem:small-t-estimate} when
$t<e^{-C\sqrt W}$, from Lemma \ref{lem:fixed-scale} and monotonicity
when $e^{-C\sqrt W}\leq t<\tau$.

Hence, for any $p>0$, we can use the layer-cake decomposition to
estimate
\begin{align*}
\mathbb E s_{min}(\mathcal L_{A_1,\ldots,A_{128}})^{-p}
&=
p\int_0^\infty t^{-p-1}
\mathbb P\left(
s_{min}(\mathcal L_{A_1,\ldots,A_{128}})\leq t
\right)\,dt\\
&\leq
p\int_0^\tau t^{c\sqrt W-p-1}\,dt
+p\int_\tau^\infty t^{-p-1}\,dt\\
&\leq p + \tau^{-p},
\end{align*}
where in the last inequality 
we took $W$ large enough to make the first 
integral at most $1$. 
\end{proof}

\subsection{Proof of Lemma \ref{lem:fixed-scale}}

For convenience, put $m_0=128$ and take $\mathcal L$ to be the operator
$$\mathcal L := \mathcal L_{A_1,...,A_{m_0}}.$$

\begin{proof}[Proof of Lemma \ref{lem:fixed-scale}]
First note that $s_{min}(\mathcal L)$ is a $2$ 
Lipshitz function of $A_1,...,A_{m_0}$, viewed as elements 
of $\mathbb{R}^{W\times W}_{sym}$ normed by the standard Hilbert-Schmidt norm. 
Thus Gaussian concentration of measure estimates, see (2.35) in \cite{Ledoux2001} imply
$$\mathbb{P}(s_{min}(\mathcal L)\leq \mathbb{E}s_{min}(\mathcal L) - t)\leq \exp\left(-\frac{Wt^2}{16}\right).$$
Thus, it suffices to prove 
$$\mathbb{E}s_{min}(\mathcal L)\geq c$$
for some $c>0$ independent of $W$.

To do this compute
$$(\mathcal L^{*}\mathcal L) (P) = 2\sum_{i=1}^{m_0}A_i^2P + 2\sum_{i=1}^{m_0}A_iP^TA_i.$$
Hence
\begin{equation}
\label{eq:expected-least-singular-value}
\mathbb{E}s_{min}(\mathcal L)^2  = \mathbb{E}s_{min}(\mathcal L^{*}\mathcal L) \geq 2\mathbb{E}(s_{min}(S) - \|\Phi-E\Phi\|_{op}),
\end{equation}
where we have defined the operators 
\begin{align*}
S(P)& := \sum_{i=1}^{m_0}A_i^2P\\
\Phi(P) &:= \sum_{i=1}^{m_0} A_iP^TA_i,
\end{align*}
and used that $\mathbb{E}\Phi$ is positive semi-definite since $A_i$ is GOE so we have
$$\mathbb{E}\langle P, \Phi(P)\rangle = \frac{m_0}{W}(\|P\|_{HS}^2 + \operatorname{Tr}(P)^2)\geq 0.$$

Intuitively, $\mathbb{E}A_i^2 =\operatorname{Id}$, 
so we expect $s_{min}(S)$ to be of order $m_0$, 
and $\Phi-\mathbb{E}\Phi$ is a sum of $m_0$ 
independent random operators, so we expect its operator 
norm to be of order $\sqrt{m_0}$ (due to sqaure root cancellation
in the random sum). One could prove this with an 
$\epsilon$-net argument, however 
we need explicity constants, and obtaining those
is quite involved. 
Instead, we invoke
strong convergence of GOE matrices, namely 
from the paper \cite{ChenGarzaVargasVanHandel2026},
which gives much sharper results much easier. 

Applying Corollary 1.2 
in \cite{ChenGarzaVargasVanHandel2026} with the polynomial 
$$p(x_1,...,x_{m_0}) = 10m_0 - \sum_{i=1}^{m_0}x_{i}^2,$$
gives that 
$$\|p(A_1,...,A_{m_0})\|_{op} \to \|p(s_1,...,s_{m_0})\|,$$
in probability as $W\to \infty$, 
where $s_1,...,s_{m_0}$ is an independent semi-circular family. 
Since $10m_0 \geq \|\sum_{i=1}^{m_0}A_i^2\|,$ with probability tending to 
$1$ as $W\to \infty$, this implies 
$$10m_0 - s_{min}(S) \to \|p(s_1,...,s_{m_0})\|,$$
in probability as $W\to \infty$. 
Computing the RHS explicity gives 
$$s_{min}(S)\to (\sqrt{m_0}-1)^2,$$
in probability as $W\to \infty$. $s_{min}(S)$ 
is uniformly integrable in $W$, and so this implies 
$$\mathbb{E}s_{min}(S)\to \left(\sqrt{m_0}-1\right)^2,$$
as $W\to\infty.$

To estimate $\mathbb{E}\|\Phi-\mathbb{E}\Phi\|_{op}$, we again invoke 
\cite{ChenGarzaVargasVanHandel2026}. If $A_1',..., A_{m_0}'$ denote 
independent GOE matrices, then writing $A\otimes B$ for the operator 
$$P\mapsto APB,$$
we have 
\begin{align*}
\mathbb{E}\|\Phi-\mathbb{E}\Phi\|_{op}
& = \mathbb{E}\|\sum_{i=1}^{m_0}A_i\otimes A_i - \mathbb{E}\sum_{i=1}^{m_0}A'_i\otimes A'_i\|_{op}\\
& \leq \mathbb{E}\|\sum_{i=1}^{m_0}A_i\otimes A_i - \sum_{i=1}^{m_0}A'_i\otimes A'_i\|_{op}\\
& = 2\mathbb{E}\|\sum_{i=1}^{m_0} A_i\otimes A_i'\|.
\end{align*}
The first line used that we can replace $P^T$ with $P$ in 
the definition of $\Phi$ without changing its operator norm. 
The second line used Jensen's inequality, and the third line
used that $A_i\pm A_i'$ has the same law as $\sqrt{2}A_i$.
Now invoking Corollary 9.6 in \cite{ChenGarzaVargasVanHandel2026} 
and arguing similarly as above gives that 
$$\mathbb{E}\|\Phi-\mathbb{E}\Phi\|_{op}\leq 2\mathbb{E}\|\sum_{i=1}^{m_0} A_i\otimes A_i'\|\to 2\|\sum_{i=1}^{m_0}s_i\otimes t_i\|,$$
as $W\to \infty$ where $s_i$ and $t_i$ are independent semi-circular families. 
The RHS can be estimated as 
$$2\|\sum_{i=1}^{m_0}s_i\otimes t_i\| \leq  4(\sqrt{m_0}+1).$$
This follows from the noncommutative Khintchine inequality, 
see Proposition 4.8 in \cite{HaagerupPisier1993}. 

Combining the above estimates with \eqref{eq:expected-least-singular-value} gives 
that for any $\epsilon>0$ we have 
$$\mathbb{E}s_{min}(\mathcal L)^2 \geq 2(1-\epsilon)(\sqrt{m_0}-1)^2 - 8(1+\epsilon)(\sqrt{m_0}+1),$$
for all $W$ sufficiently large, depending on $\epsilon$.
One can check that taking $\epsilon$ sufficiently small and using that $m_0 = 128$ implies that 
$$\mathbb{E}s_{min}(\mathcal L)^2 \geq 1,$$
for all $W$ sufficiently large. 

To finish note that Cauchy-Schwarz,
operator norm estimate on GOE matriecs, 
and the estimate $s_{min}(\mathcal L)\leq 2\sum_{i=1}^{m_0}\|A_i\|_{op}$
implies 
$$1\leq (\mathbb{E}s_{min}(\mathcal L)^2 )^2 \leq (\mathbb{E} s_{min}(\mathcal L))(\mathbb{E} s_{min}(\mathcal L)^3)\leq 1000 m_0\mathbb{E} s_{min}(\mathcal L).$$
Hence $\mathbb{E}s_{min}(\mathcal L)\geq c $ for some $c>0$ independent of $W$. 
\end{proof}

\subsection{Proof of Lemma \ref{lem:small-t-estimate}}

For convenience, put $m_0=128$ and take $\mathcal L$ to be the operator
$$\mathcal L := \mathcal L_{A_1,...,A_{m_0}}.$$

\begin{proof}[Proof of Lemma \ref{lem:small-t-estimate}]
We take $W$ sufficiently large at various points in the proof. Let
$$
M_0
=
C\left(
\frac{\log(1/t)}{\sqrt W}
\right)^{1/2},
$$
and define the event
$$
\Omega:=\{\|\mathcal L\|_{op}\leq M_0\}.
$$
To estimate the probability $s_{min}(\mathcal L)\leq t$, we split the
unit sphere in
$\mathbb R^{W\times W}$ into pieces as follows. For any
$\mathbf k=(k_1,\ldots,k_r)\in\mathbb Z^r$, define
$$
Q_{\mathbf k}
:=
\left\{
P\in\mathbb R^{W\times W}:\|P\|_{HS}=1,
\ \operatorname{rank}(P)=r,
\ \sigma_i(P)\in (2^{-k_i-1},2^{-k_i}]
\right\}.
$$
Then we can estimate
\begin{equation}
\label{eq:first-reduction}
\begin{aligned}
\mathbb P\left(s_{min}(\mathcal L)\leq t\right)
&=
\mathbb P\left(
\inf_{\|P\|_{HS}=1}\|\mathcal L(P)\|_{HS}\leq t
\right)\\
&\leq
\mathbb P(\Omega^c) +
\sum_{\mathbf k}
\mathbb P\left(
\Omega\cap
\left\{
\inf_{P\in Q_{\mathbf k}}\|\mathcal L(P)\|_{HS}\leq2t
\right\}
\right).
\end{aligned}
\end{equation}
Here the last summation is over $\mathbf k = (k_1,...,k_r)\in \mathbb{N}^{r}$ with $1\leq r\leq W$,
$Q_{\mathbf k}\ne\varnothing$,
and $k_r\leq\lceil\log_2(CM_0\sqrt W/t)\rceil$. We can restrict 
to this set of $\mathbf{k}$ using the bound on $\|\mathcal L\|_{op}$ from $\Omega$, 
and truncating
any $P\in\mathbb R^{W\times W}$ to its singular values larger than
$t/(CM_0\sqrt W)$. 

The first term in \eqref{eq:first-reduction} can be estimated using
standard GOE operator-norm estimates as
\begin{equation}
\label{eq:GOE-operator-norm-estimate}
\mathbb P(\Omega^c)
\leq
m_0\mathbb P\left(
\|A_1\|_{op}>\frac{M_0}{2\sqrt{m_0}}
\right)
\leq
Ce^{-cWM_0^2}\leq Ct^{cW}.
\end{equation}

We estimate the probabilities in the sum in \eqref{eq:first-reduction} using an
$\epsilon$-net versus small-ball probability argument.
The entropy of $Q_{\mathbf k}$
is captured by the following
quantity. For any $\mathbf k=(k_1,\ldots,k_r)$ and $\epsilon>0$,
define
$$
H(\mathbf k,\epsilon)
:=
3W\sum_{i=1}^r
\log_+\left(
\frac{C_0 2^{-k_i}\sqrt r}{\epsilon}
\right).
$$
Here $C_0$ is a sufficiently large absolute constant. 
The following bounds the entropy of $Q_{\mathbf k}$ at scale
$\epsilon$ using $H(\mathbf k,\epsilon)$.
\begin{claim}
\label{claim:profile-entropy}
For any $\mathbf k$, with $Q_{\mathbf k} \neq \emptyset$, and $\epsilon>0$, there exists an $\epsilon$-net
$\mathcal N\subset Q_{\mathbf k}$ in the $\|\cdot\|_{HS}$ norm 
satisfying
$$
\log |\mathcal N|
\leq
H(\mathbf k,\epsilon).
$$
\end{claim}

\begin{proof}[Proof of Claim A.4]
Let $\mathbf k = (k_1,...,k_r)$, 
and take $\epsilon>0$. We construct 
our $\epsilon$-net through 
singular value decomposition. There exists a subset
of the orthogonal matrices $\mathcal N_1 \subset O(W)\subset \mathbb{R}^{W\times W}$ 
such that for each $U\in O(W)$ there is a $\tilde{U}\in \mathcal N_1$
such that for $i = 1,...,r$ the $i$-th columns of $\tilde{U}$
and $U$ are  within $2^{k_i}\epsilon/(3\sqrt{r})$ in $\|\cdot \|_{2}$ norm. 
By the usual volume counting argument, we can find such a set with cardinality at most 
$$\exp\left(\sum_{i=1}^{r}W\log_{+}(C2^{-k_i}\sqrt{r}/\epsilon)\right).$$
Now take a set $\mathcal N_2 \subset \mathbb{R}^{W\times W},$
of diagonal matrices,
such that for any diagonal $D\in \mathbb{R}^{W\times W}$
with $D_{ii}\in [2^{-k_i-1}, 2^{-k_i}]$ for $i = 1,...,r$
$D_{jj} = 0$ for $j>r$ and $\sum_{i=1}^{r} D_{ii}^2 = 1$,
there exists $\tilde{D}\in \mathcal N_2$, satisfying the same constraints, 
with 
$\sup_{i=1,...,r}|D_{ii}-\tilde{D}_{ii}|\leq \epsilon/(3\sqrt{r}).$
We can find such a set with cardinality at most
$$\exp\left(\sum_{i=1}^{r}\log_{+}(C2^{-k_i}\sqrt{r}/\epsilon)\right).$$
Now note that, by taking a singular value decomposition of $P$, 
$$\mathcal N := \{\tilde{U}\tilde{D}\tilde{V}:\tilde{U}, \tilde{V}^T\in \mathcal N_1, \tilde D\in \mathcal N_2\}\subset Q_{\mathbf k}$$
is an $\epsilon$-net for $Q_{\mathbf k}$ in the HS-norm, since for any $U,V,\tilde U, \tilde V\in O(W)$
$$\|UDV - \tilde U \tilde D \tilde V\|_{HS} \leq \|(U-\tilde U)D\|_{HS} + \|(D - \tilde D)\|_{HS}  + \|\tilde D (V-\tilde V)\|_{HS}.$$
Using the cardinality estimates above then gives
$$\log |\mathcal N|\leq (2W+1)\sum_{i=1}^{r} \log_{+}(C2^{-k_i}\sqrt{r}/\epsilon).$$
\end{proof}

For any $\mathbf k$, the
probability that $\mathcal L(P)$ lies in the $\epsilon$-ball for
$P\in Q_{\mathbf k}$
is measured by the following quantity,
$$
S(\mathbf k,\epsilon)
:=
-\log\sup_{P\in Q_{\mathbf k}}
\mathbb P\left(\|\mathcal L(P)\|_{HS}\leq\epsilon\right).
$$
Note that for any $\epsilon\geq2t$ and $\mathbf k$, 
Claim \ref{claim:profile-entropy} and
a union bound give
\begin{equation}
\label{eq:entropy-vs-prob}
\mathbb P\left(
\Omega\cap
\left\{
\inf_{P\in Q_{\mathbf k}}\|\mathcal L(P)\|_{HS}\leq2t
\right\}
\right)
\leq
\exp\left(
H(\mathbf k,\epsilon)
-S(\mathbf k,(M_0+1)\epsilon)
\right).
\end{equation}

Using Lemma \ref{lem:LP-singular-value-amplification}, we have the
following estimate on $S(\mathbf k,\epsilon)$.
\begin{claim}
\label{claim:small-ball-estimate}
For any $\epsilon>0$ and $\mathbf k=(k_1,\ldots,k_r)$, we have
\begin{equation}
\label{eq:S-profile-bound}
S(\mathbf k,\epsilon)
\geq
\frac{m_0W}{25}
\max_{1\leq h\leq r}
\left[
\sum_{i=1}^h
\log\left(
\frac{2^{-k_i}\sqrt h}{C\epsilon}
\right)
\right]_+.
\end{equation}
\end{claim}

Heuristically, Lemma A.6 below implies that for each
$\sigma_i(P)$, 
$\mathcal L$ has $m_0W/25$ singular directions of comparable size. 
Hence each $\sigma_i(P)$
reduces the small ball probability roughly by the factor 
$$\left(\frac{C\epsilon}{\sigma_i(P)}\right)^{m_0 W/25}$$

\begin{proof}[Proof of Claim] 
Fix $P\in Q_{\mathbf k}$ and define 
$L_P:\mathbb{R}^{W\times W}_{sym} \to \mathbb{R}^{W\times W}_{sym}$ 
by 
$$L_P(A) = AP + P^TA.$$ 
For any $N\in [1,\operatorname{rank}(L_P)]$,
the projection of $L_P(A)$ onto the top $N$ 
left singular vectors of $L_P$
is a centered Gaussian in $\mathbb{R}^{N}$ with covariance 
$$\frac{2}{W}\operatorname{diag}(\sigma_1(L_P)^2, ..., \sigma_{N}(L_P)^2).$$
The standard Gaussian small ball estimate thus implies 
$$\mathbb{P}(\|L_P(A_1)\|_{HS}\leq \epsilon)\leq 
\left(\frac{C\epsilon}{\sqrt{N}}\right)^{N}\left(\frac{(C\sqrt{W})^N}{\prod_{i=1}^{N}\sigma_i(L_P)}\right) = \frac{(C\epsilon \sqrt{W/N})^{N}}{\prod_{i=1}^{N}\sigma_i(L_P)}.$$
The first factor bounds the volume of the $\epsilon$-ball in $\mathbb{R}^{N}$
while the second factor uses that a Gaussian 
with covariance matrix $\Sigma\in \mathbb{R}^{N\times N}$ has density bounded above by 
$C^N/\sqrt{\operatorname{det(\Sigma)}}$. 
Since $A_1,...,A_{m_0}$ are independent
and $\mathcal L(P) = (L_P(A_i))_{i=1}^{m_0}$ we have 
\begin{align*}
\log \mathbb{P}\left(\|\mathcal L(P)\|_{HS}\leq \epsilon\right)
& \leq 
m_0\log \left( \frac{(C\epsilon \sqrt{W/N})^{N}}{\prod_{i=1}^{N}\sigma_i(L_P)}\right)\\
& = m_0\sum_{i=1}^{N} \log \frac{C\epsilon \sqrt{W}}{\sqrt{N}\sigma_i(L_P)}
\end{align*}
Now Lemma \ref{lem:LP-singular-value-amplification}
lets us estimate the $\sigma_i(L_P)$. It  
implies that 
$$\sigma_{1}(L_P),...,\sigma_{\lfloor W/24\rfloor}(L_P)\geq c 2^{-k_1},$$
and 
$$\sigma_{\lfloor W/24\rfloor+1}(L_P),...,\sigma_{2\lfloor W/24\rfloor}(L_P)\geq c 2^{-k_2},$$
and so on. Thus for any $h= 1,...,\operatorname{rank}(P),$ 
combining this with the above and taking $N = h\lfloor W/24\rfloor$ gives 
\begin{align*}
\log \mathbb{P}\left(\|\mathcal L(P)\|_{HS}\leq \epsilon\right)
& \leq \frac{m_0 W}{25}\sum_{i=1}^{h} \log \frac{C\epsilon}{\sqrt{h}2^{-k_i}}. 
\end{align*}
This holds for any $P\in Q_{\mathbf k}$ and $h = 1,..., \operatorname{rank}(P) = r,$
so the claim follows by negating both sides, 
and noting that we always have the bound $S(k,\epsilon)\geq 0$
\end{proof}

At this point we would like to apply \eqref{eq:entropy-vs-prob} with
$\epsilon=t$. However, this does not work 
for all $\mathbf k$. The issue is that a $P\in Q_{\mathbf k}$ may have many singular values 
at scale $\epsilon = t$. These singular values increase the entropy $H(\mathbf k, \epsilon)$, 
but may not improve the small-ball estimate
since it has radius $(M_0+1)\epsilon\gg \epsilon$. 
Hence, we need to 
apply \eqref{eq:entropy-vs-prob} at a scale $\epsilon$
for which most of the entropy comes from singular values 
at a scale $\gg\epsilon$. \eqref{eq:entropy-change}
captures this. To find such a scale
we use a pigeon hole princple argument.

For any $\mathbf k$ with $Q_{\mathbf k}\neq \emptyset$, 
we find an
integer $j\in [0, \lfloor \frac{3}{4}\log_2(1/t)\rfloor]$ such that
$$
S(\mathbf k,(M_0+1)2^jt)-H(\mathbf k,2^jt)
\geq
\frac12H(\mathbf k,2^jt).
$$
To do this note that
for all $j\in [0,\frac{3}{4}\log_2(1/t)]$
$$
cW\log(1/t)\leq H\left(
\mathbf k,
2^{j}t
\right)
\leq
CW^2\log(1/t),
$$
for some absolute constants $c,C>0$. The upper bound 
follows since $H$ has at most $W$ summands, and 
the lower bound because $Q_{\mathbf k}$ is non-empty 
so $1\leq \sum 2^{-2k_i}$
and $2^j t$ is small enough that 
$H$ must have at least $1$ summand larger than $c\log(1/t)$. 
This implies 
$$\prod_{j=1}^{\lfloor \frac{3}{4}\log_2(1/t)\rfloor}\frac{H(\mathbf k, 2^{j-1}t)}{H(\mathbf k, 2^{j}t)} = \frac{H(\mathbf k, t)}{H(\mathbf k, 2^{\lfloor \frac{3}{4}\log_2(1/t)\rfloor}t)}\leq CW.$$
Hence, there an integer exists $1\leq j\leq \lfloor\frac34\log_2(1/t)\rfloor$ satisfying
\begin{equation}
\label{eq:entropy-change}
H(\mathbf k,2^{j-1}t)
-H(\mathbf k,2^jt)
\leq
\frac{C\log W}{\log(1/t)}
H(\mathbf k,2^jt),
\end{equation}
otherwise the product above would be too large. 
If $h$ is the number of terms contributing to $H(\mathbf k,2^jt)$
then each of these terms contributes $3W\log2$ to the difference on the
left, so we have 
$$
Wh\frac{c\log(1/t)}{\log W}\leq H(\mathbf k,2^jt).
$$
Thus, using the definitions of $S$ and $H$, we may estimate 
\begin{align*}
&S(\mathbf k,(M_0+1)2^jt)-H(\mathbf k,2^jt)\\
&\quad\geq
\frac{m_0}{75}
\left(
H(\mathbf k,2^jt)
-CWh\left(1+\log M_0+\log\frac rh\right)
\right)
-H(\mathbf k,2^jt)\\
&\quad\geq
\frac12H(\mathbf k,2^jt)\\
&\quad\geq
cW\log(1/t).
\end{align*}
The first inequality follows by comparing the definition of $H$ with
the first $h$ terms in \eqref{eq:S-profile-bound}. The second uses the
preceding lower bound on $H$, taking $W$ sufficiently large, and $m_0=128$, while the last uses the monotonicity
of $H$ and its lower bound noted above.
It follows from \eqref{eq:entropy-vs-prob} that
\begin{equation}
\label{eq:probability-estimate} 
\mathbb P\left(
\Omega\cap
\left\{
\inf_{P\in Q_{\mathbf k}}\|\mathcal L(P)\|_{HS}\leq2t
\right\}
\right)
\leq
\exp\left(-cW\log(1/t)\right).
\end{equation}

Finally, since the sum in \eqref{eq:first-reduction} is over $\mathbf k$ satisfying
$-1\leq k_i\leq\lceil\log_2(CM_0\sqrt W/t)\rceil$, 
the number of terms in the sum is at most
$$
W\left(
C\left\lceil\log_2\left(\frac{CM_0\sqrt W}{t}\right)\right\rceil
\right)^W
\leq
\exp\left(CW\log\bigl(2+\log(1/t)\bigr)\right). 
$$
Combining with \eqref{eq:probability-estimate}
and summing over the admissible
profiles in \eqref{eq:first-reduction}, gives
\begin{align*}
\mathbb P\left(s_{min}(\mathcal L)\leq t\right)
&\leq
Ct^{c\sqrt W}
+\exp\left(-cW\log(1/t)\right)\\
&\leq
(Ct)^{c\sqrt W}.
\end{align*}
\end{proof}

\subsection{Auxiliary estimate for Lemma \ref{lem:small-t-estimate}}
We prove the singular-value estimate used in the proof
of Claim \ref{claim:small-ball-estimate} above

\begin{lemma}
\label{lem:LP-singular-value-amplification}
For $P\in \mathbb{R}^{W\times W}$, define
$L_P:\mathbb{R}^{W\times W}_{sym}\to \mathbb{R}^{W\times W}_{sym}$
by
$$
L_P(M)=MP+P^TM. 
$$
Then, for every $s>0$,
$$
\left|\left\{i:\sigma_i(L_P)\geq \frac{s}{4\sqrt{2}}\right\}\right|
\geq
\frac{W}{24}
\left|\left\{j:\sigma_j(P)\geq s\right\}\right|.
$$
\end{lemma}

We note that if $L_P$ were simply $M\mapsto MP$ this lemma, 
would be immediate. However, the danger is that 
the image of the symmetric matrices under $P$ 
is very close to the anti-symmetric matrices 
so that $\|MP + P^TM\|_{HS}\ll \|MP\|_{HS}$. Claim 
A.7 below is used to rule this out. 

\begin{proof}[Proof of Lemma \ref{lem:LP-singular-value-amplification}]
Fix $P\in\mathbb{R}^{W\times W}$, $s>0$, and define 
$$k:= \left|\left\{i: \sigma_{i}(P)\geq s\right\}\right|,$$
The idea is to construct a subspace of dimension at least $Wk/24$, 
such that for all $M$ in that subspace 
$$\|L_P(M)\|_{HS} \geq \frac{s}{4\sqrt{2}}\|M\|_{HS}.$$

Suppose $\operatorname{rank}(P) \leq W/2$. 
Let 
$$P:= \sum_{i=1}^{\operatorname{rank}(P)}\sigma_{i}u_iv_i^T,$$
be a singular value decomposition of $P$ wiht $\sigma_1\geq ...\geq \sigma_{\operatorname{rank}(P)}$. 
Since $\operatorname{rank}(P)\leq W/2$ there exists orthonormal vectors 
$w_1,...,w_{\rfloor W/2\lfloor}\in \mathbb{R}^{W}$, such that 
$$\langle w_i, u_j\rangle = 0$$
for all $i = 1,...,W/2$ and $j= 1,..., \operatorname{rank}(P).$
Now consider 
$$V= \operatorname{span}(w_ju_i^{T}+ u_iw_j^T)_{i\in [1,k], j\in [1,W/2]}.$$
For any $M= \sum_{j=1}^{W/2}\sum_{i=1}^{k}m_{ij}(w_ju_i^{T}+ u_iw_j^T)\in V,$
we can compute 
\begin{align*}
\|MP\|_{HS}^2
& = \left\|\sum_{j=1}^{W/2}\sum_{i=1}^{k}m_{ij}(w_ju_i^{T}+ u_iw_j^T)P\right\|_{HS}^2\\
& = \left\|\sum_{j=1}^{W/2}\sum_{i=1}^{k}m_{ij}\sigma_{i}w_jv_i^{T}\right\|_{HS}^2\\
& = \sum_{j=1}^{W/2}\sum_{i=1}^{k}m_{ij}^2\sigma_i^2\\
& \geq \frac{s^2}{2}\|M\|_{HS}^2.
\end{align*}
In the second equality we used the $w_j$ are orthogonal to the $u_i$, and 
that $u_i^TP = \sigma_{i}v_i^T$. In the third equality we used that $w_jv_i^T$ are 
orthonormal, and to pass to the last line we used that $\sigma_i\geq s$ for all $i\in[1,k]$, 
and that the $w_ju_i^T+ u_iw_j^T$ are orthogonal and have HS-norm $\sqrt{2}$. 

Now we find a large subpace of $V'$ of $V$ such that 
for any $M\in V'$ we have 
$$\|MP + P^TM\|_{HS} \geq \|MP\|_{HS}.$$
For this we use the following lemma. 
The idea is that the kernel of $B+B^T$ is 
the anti-symmetric matrices. The condition below
ensures that the span of $B_1,...,B_N$ doesn't 
have large intersection with the space of anti-symmetric 
matrices.

\begin{claim}
If $B_1,...,B_N\in \mathbb{R}^{W\times W}$ are orthonormal
with respect to the standard Hilbert Schmidt inner product
$\langle A, B\rangle_{HS} = \operatorname{Tr}(AB^T)$, 
and 
$$\sum_{i=1}^{N} \operatorname{Tr}(B_i^2)\geq 0,$$
then there exists a subspace $V'\subset \operatorname{span}(B_1,...,B_N)$ 
with $\operatorname{dim}(V')\geq N/3$ such that 
for any $B\in V'$
$$\|B+B^T\|_{HS} \geq \|B\|_{HS}.$$
\end{claim}
\begin{proof}
Let $V:= \operatorname{span} (B_1,...,B_N)$, and consider the 
orthogonal projection $\Pi_{V}:V\to \mathbb{R}^{sym}_{W\times W}$ 
given by 
$$\Pi_V(M) = \frac{M + M^T}{2}.$$
Using the ONB for $V$ given by $B_1,...,B_N$ we have 
\begin{align*}
\|\Pi_{V}\|_{HS(V,\mathbb{R}^{W\times W}_{sym})}^2 
& = \sum_{i\in [1,N]} \|\Pi_{V}(B_i)\|_{HS}^2 \\
& = \sum_{i\in [1,N]}\|\frac{B_i + B_i^T}{2}\|_{HS}^2\\
& = \frac{1}{2}\sum_{i\in [1,N]} 1 + \operatorname{Tr}(B_i^2)\\
& \geq \frac{N}{2}
\end{align*}
Hence, $\sum_{i=1}^{N}\sigma_i(\Pi_V)^2 = \|\Pi_{V}\|_{HS}^2\geq N/2$. 
But also for all $i\in [1,N],$ we have $\sigma_i(\Pi_V)\leq \|\Pi_V\|_{op}\leq 1,$
so $\Pi_{V}$ has at least $N/3$ singular values, larger than $1/2$.
If $V'$ is the subspace spanned by the right singular directions 
corresponding to those singular values we have 
$$\|B^T + B\|_{HS} = 2 \|\Pi_{V}(B)\|_{HS} \geq \|B\|_{HS},$$
for any $B\in V'$. So the claim follows. 
\end{proof}

We apply this to the subspace $VP \subset \mathbb{R}^{W\times W}$. 
As noted above $(w_jv_i^T)_{j\in [1,W/2], i\in [1,k]}$, 
is an orthonormal basis for $VP$. Moreover, we have 
\begin{align*}
\sum_{j=1}^{W/2}\sum_{i=1}^k\operatorname{Tr}\left((w_jv_i^T)(w_jv_i^T)\right) 
& = \sum_{j=1}^{W/2}\sum_{i=1}^k \langle v_i,w_j\rangle_{\mathbb{R}^{W}}^2\geq 0.
\end{align*}
Since $\operatorname{dim}(V) = \operatorname{dim}(VP) = Wk/2,$
there exists a subspace $V'\subset V$ with $\operatorname{dim}(V')\geq Wk/6$
such that 
\begin{align*}
\inf_{M\in V', \|M\|_{HS} = 1}\|L_P(M)\|_{HS}
& = \inf_{M\in V', \|M\|_{HS} = 1}\|MP + P^TM\|_{HS}\\
& \geq \inf_{M\in V', \|M\|_{HS} = 1}\|MP\|_{HS}\\
& \geq \frac{s}{\sqrt{2}}. 
\end{align*}
Applying the max-min characterization of singular values
to $L_P$ proves the lemma when $\operatorname{rank}(P)\leq W/2$. 

Now we prove that regardless of the rank of $P$ we always have
\begin{equation}
\label{eq:rank-free-estimate}
\left|\left\{i:\sigma_i(L_P)\geq s\right\}\right|
\geq
\frac{k(k+1)}{6}.
\end{equation}
and together with the above estimate this will imply the lemma. 
For this, define the subspace
$$V := \left\{\sum_{i,j\in [1,k]}m_{ij}u_iu_j^T : m_{ij} = m_{ji}\right\}.$$
For any $M = \sum_{i,j\in [1,k]} m_{ij}u_iu_j^T \in V,$
we can compute 
\begin{align*}
\|MP\|_{HS}^2
& = \left\|\sum_{i,j\in [1,k]}m_{ij}u_iu_j^{T}P\right\|_{HS}^2\\
& = \left\|\sum_{i,j\in [1,k]}m_{ij}\sigma_{j}u_iv_j^{T} \right\|_{HS}^2\\
& \geq s^2 \sum_{i,j\in [1,k]}m_{ij}^2\\
& \geq s^2\|M\|_{HS}^2.
\end{align*}

To pass to the third line we used that $\sigma_j\geq s$ for 
$j\in [1,k]$ and that $u_iv_j^T$ are orthonormal in the HS-norm.

To apply Claim A.7 as above, note that the vectors 
$E_{ii} = u_iv_i^T$
for $1\leq i\leq k$ together with the vectors
$$
E_{ij} = \frac{
\sigma_j u_iv_j^T+\sigma_i u_jv_i^T
}{
\sqrt{\sigma_i^2+\sigma_j^2}
}
$$
for $1\leq i<j\leq k$ form an orthonormal basis
for $VP$. Hence 
\begin{align*}
\sum_{i=1}^k\operatorname{Tr}(E_{ii}^2)
+\sum_{1\leq i< j\leq k} \operatorname{Tr}(E_{ij}^2) 
& = \sum_{i\in [1,k]} \langle u_i,v_i\rangle^2 + \sum_{1\leq i< j\leq k} \frac{\sigma_j^2\langle u_i,v_j\rangle^2 + \sigma_i^2\langle u_j,v_i\rangle^2 + 2\sigma_i\sigma_j
\langle u_i,v_i\rangle
\langle u_j,v_j\rangle
}{
\sigma_i^2+\sigma_j^2
}\\
& \geq \sum_{i\in [1,k]}\langle u_i,v_i\rangle^2 + \sum_{1\leq i<j\leq k}\frac{2\sigma_i\sigma_j}{\sigma_i^2+\sigma_j^2}\langle u_i,v_i\rangle\langle u_j, v_j\rangle.\\
& \geq 0.
\end{align*}

The second line follows by throwing away the first two positive 
terms in the numerator in the second summation. For the last line, we use that
the matrix $\left(\frac{\sigma_i\sigma_j}{\sigma_i^2+\sigma_j^2}\right)_{i,j\in[1,k]}$
is positive semidefinite. This follows from the positivity of the Cauchy matrix
$\left(\frac{1}{\sigma_i^2+\sigma_j^2}\right)_{i,j\in[1,k]}$.

Thus there exists a subspace 
$V'\subset V$ such that $\operatorname{dim}(V') \geq \operatorname{dim}(V)/3\geq \frac{k(k+1)}{6}$
such that
\begin{align*}
\inf_{M\in V', \|M\|_{HS} = 1} \|L_P(M)\|_{HS} \geq s.
\end{align*}
Now applying max-min as above implies \eqref{eq:rank-free-estimate}. 

To finish note that 
if $P$ has at least $W/2$ singular values 
larger than $s/4\sqrt{2}$ then the above implies 
$$
\left|\left\{i:
\sigma_i(L_P)\geq \frac{s}{4\sqrt{2}}
\right\}\right|
\geq
\frac{
(W/2)(W/2+1)
}{6}\geq \frac{Wk}{24}
$$
which implies the Lemma. 

Otherwise if $P$ has $\leq W/2$ 
singular values $> s/4\sqrt{2}$, let $Q$ be the truncation of $P$ to its first
$\lfloor W/2\rfloor$ singular values. We must have 
$$
\|P-Q\|_{\mathrm{op}}
\leq
\frac{s}{4\sqrt{2}},
$$
and moreover $Q$ must have at least
$k$ singular values larger than $s$. 
Hence, the first estimate proved above applied 
to $Q$ yields
$$
\left|\left\{i:
\sigma_i(L_Q)\geq \frac{s}{\sqrt{2}}
\right\}\right|
\geq
\frac{Wk}{6}.
$$
But since $\|L_P-L_Q\|_{\mathrm{op}}\leq 2\|P-Q\|_{op}\leq \frac{s}{2\sqrt{2}},$
the same estimate follows for $P$, with 
$\frac{s}{\sqrt{2}}$ replaced by $\frac{s}{2\sqrt{2}}$, 
which implies the lemma. 
\end{proof}

\nocite{Forrester1,Forrester2,ForresterZhang2020}
\printbibliography

@article{ChenGarzaVargasVanHandel2026,
  author       = {Chen, Chi-Fang and Garza-Vargas, Jorge and van Handel, Ramon},
  date         = {2026},
  doi          = {10.1007/s00039-026-00744-2},
  journaltitle = {Geom. Funct. Anal.},
  title        = {A new approach to strong convergence {II}: The classical ensembles},
}

@article{HaagerupPisier1993,
  author       = {Haagerup, Uffe and Pisier, Gilles},
  date         = {1993},
  journaltitle = {Duke Math. J.},
  number       = {3},
  pages        = {889--925},
  title        = {Bounded linear operators between {$C^*$}-algebras},
  volume       = {71},
}

@article{BinderGoldsteinVoda2015,
  author       = {Binder, Ilia and Goldstein, Michael and Voda, Mircea},
  date         = {2015},
  journaltitle = {J. Spectral Theory},
  number       = {2},
  pages        = {355--395},
  title        = {On fluctuations and localization length for the {A}nderson model on a strip},
  volume       = {5},
}

@article{RudelsonVershynin2009,
  author       = {Rudelson, Mark and Vershynin, Roman},
  date         = {2009},
  doi          = {10.1002/cpa.20294},
  journaltitle = {Comm. Pure Appl. Math.},
  number       = {12},
  pages        = {1707--1739},
  title        = {The smallest singular value of a random rectangular matrix},
  volume       = {62},
}

@online{TaoLeastSingularValue2010,
  author  = {Tao, Terence},
  url     = {https://terrytao.wordpress.com/2010/03/05/254a-notes-7-the-least-singular-value/},
  date    = {2010-03-05},
  title   = {{254A}, Notes 7: The least singular value},
  urldate = {2026-07-23},
}

@misc{QuasDewitt,
  author      = {Bednarski, Sam and DeWitt, Jonathan and Quas, Anthony},
  date        = {2025},
  doi         = {10.48550/arXiv.2507.04058},
  eprint      = {2507.04058},
  eprintclass = {math.PR},
  eprinttype  = {arXiv},
  title       = {Effective gaps between singular values of non-stationary matrix products subject to non-degenerate noise},
}

@article{Bourgain2012,
  author       = {Bourgain, Jean},
  date         = {2012},
  journaltitle = {J. Anal. Math.},
  pages        = {273--286},
  title        = {On the {F}urstenberg measure and density of states for the {A}nderson--{B}ernoulli model at small disorder},
  volume       = {117},
}

@article{BourgainGamburd2008,
  author       = {Bourgain, Jean and Gamburd, Alex},
  date         = {2008},
  journaltitle = {Ann. of Math.},
  number       = {2},
  pages        = {625--642},
  title        = {Uniform expansion bounds for {C}ayley graphs of $\mathrm{SL}_2(\mathbb{F}_p)$},
  volume       = {167},
}

@misc{BourgainGamburd2012,
  author      = {Bourgain, Jean and Gamburd, Alex},
  date        = {2012},
  eprint      = {1108.6264},
  eprintclass = {math.GR},
  eprinttype  = {arXiv},
  title       = {A spectral gap theorem in $\mathrm{SU}(d)$},
}

@book{BougerolLacroix1985,
  author    = {Bougerol, Philippe and Lacroix, Jean},
  location  = {Boston},
  publisher = {Birkh\"{a}user},
  date      = {1985},
  series    = {Progress in Probability and Statistics},
  title     = {Products of Random Matrices with Applications to {S}chr\"{o}dinger Operators},
  volume    = {8},
}

@article{CoopermanRowan2025,
  author       = {Cooperman, William and Rowan, Keefer},
  date         = {2025},
  doi          = {10.1007/s00222-025-01384-3},
  journaltitle = {Invent. Math.},
  number       = {3},
  pages        = {863--959},
  title        = {Exponential scalar mixing for the 2D {N}avier--{S}tokes equations with degenerate stochastic forcing},
  volume       = {243},
}

@article{Furstenberg1963,
  author       = {Furstenberg, Harry},
  date         = {1963},
  journaltitle = {Trans. Amer. Math. Soc.},
  pages        = {377--428},
  title        = {Noncommuting random products},
  volume       = {108},
}

@article{FurstenbergKesten1960,
  author       = {Furstenberg, Harry and Kesten, Harry},
  date         = {1960},
  journaltitle = {Ann. Math. Statist.},
  pages        = {457--469},
  title        = {Products of random matrices},
  volume       = {31},
}

@article{FurstenbergKifer1983,
  author       = {Furstenberg, Harry and Kifer, Yuri},
  date         = {1983},
  doi          = {10.1007/BF02760620},
  journaltitle = {Israel J. Math.},
  number       = {1--2},
  pages        = {12--32},
  title        = {Random matrix products and measures on projective spaces},
  volume       = {46},
}

@article{FyodorovMirlin1991,
  author       = {Fyodorov, Yan V. and Mirlin, Alexander D.},
  date         = {1991},
  journaltitle = {Phys. Rev. Lett.},
  number       = {18},
  pages        = {2405--2409},
  title        = {Scaling properties of localization in random band matrices: {A} $\sigma$-model approach},
  volume       = {67},
}

@article{Forrester1,
  author       = {Forrester, Peter J.},
  date         = {2013},
  doi          = {10.1007/s10955-013-0735-7},
  journaltitle = {J. Stat. Phys.},
  pages        = {796--808},
  title        = {{L}yapunov exponents for products of complex {G}aussian random matrices},
  volume       = {151},
}

@article{Forrester2,
  author       = {Forrester, Peter J.},
  date         = {2015},
  doi          = {10.1088/1751-8113/48/21/215205},
  journaltitle = {J. Phys. A: Math. Theor.},
  number       = {21},
  pages        = {215205},
  title        = {Asymptotics of finite system {L}yapunov exponents for some random matrix ensembles},
  volume       = {48},
}

@article{ForresterZhang2020,
  author       = {Forrester, Peter J. and Zhang, Jiyuan},
  date         = {2020},
  doi          = {10.1007/s10955-019-02474-2},
  journaltitle = {J. Stat. Phys.},
  pages        = {558--575},
  title        = {{L}yapunov exponents for some isotropic random matrix ensembles},
  volume       = {180},
}

@article{GorodetskiKleptsyn2026,
  author       = {Gorodetski, Anton and Kleptsyn, Victor},
  date         = {2026},
  doi          = {10.56994/JAMR.004.001.001},
  journaltitle = {Journal of the Association for Mathematical Research},
  number       = {1},
  pages        = {1--41},
  title        = {Non-stationary version of {F}urstenberg's theorem on random matrix products},
  volume       = {4},
}

@book{Hall2015,
  author    = {Hall, Brian C.},
  publisher = {Springer},
  date      = {2015},
  doi       = {10.1007/978-3-319-13467-3},
  edition   = {2},
  series    = {Graduate Texts in Mathematics},
  title     = {Lie Groups, Lie Algebras, and Representations: An Elementary Introduction},
  volume    = {222},
}

@article{LindenstraussVarju2016,
  author       = {Lindenstrauss, Elon and Varj\'{u}, P\'{e}ter P.},
  date         = {2016},
  doi          = {10.1215/00127094-3167490},
  journaltitle = {Duke Math. J.},
  number       = {6},
  pages        = {1061--1127},
  title        = {Random walks in the group of {E}uclidean isometries and self-similar measures},
  volume       = {165},
}

@article{Kittle2025,
  author       = {Kittle, Samuel},
  date         = {2025},
  doi          = {10.1112/plms.70072},
  eid          = {e70072},
  journaltitle = {Proc. Lond. Math. Soc.},
  number       = {1},
  title        = {Absolutely continuous {F}urstenberg measures},
  volume       = {131},
}

@article{Kogler2025,
  author       = {Kogler, Constantin},
  date         = {2025},
  doi          = {10.1007/s11854-025-0369-0},
  journaltitle = {J. Anal. Math.},
  number       = {2},
  pages        = {743--798},
  title        = {Local limit theorem for random walks on symmetric spaces},
  volume       = {155},
}

@article{Hairer2011,
  author       = {Hairer, Martin},
  date         = {2011},
  doi          = {10.1016/j.bulsci.2011.07.007},
  journaltitle = {Bull. Sci. Math.},
  number       = {6--7},
  pages        = {650--666},
  title        = {On Malliavin's proof of {H}\"{o}rmander's theorem},
  volume       = {135},
}

@article{Hormander1967,
  author       = {H\"{o}rmander, Lars},
  date         = {1967},
  doi          = {10.1007/BF02392081},
  journaltitle = {Acta Math.},
  pages        = {147--171},
  title        = {Hypoelliptic second order differential equations},
  volume       = {119},
}

@article{IsopiNewman1992,
  author       = {Isopi, Marco and Newman, Charles M.},
  date         = {1992},
  journaltitle = {Comm. Math. Phys.},
  pages        = {591--598},
  title        = {The triangle law for {L}yapunov exponents of large random matrices},
  volume       = {143},
}

@article{BreuillardGamburd2010,
  author       = {Breuillard, Emmanuel and Gamburd, Alex},
  date         = {2010},
  journaltitle = {Geom. Funct. Anal.},
  number       = {5},
  pages        = {1201--1209},
  title        = {Strong uniform expansion in $\mathrm{SL}(2,p)$},
  volume       = {20},
}

@article{Shapiro,
  author       = {Shapiro, Jacob},
  date         = {2024},
  doi          = {10.1007/s11005-024-01796-x},
  eid          = {47},
  journaltitle = {Lett. Math. Phys.},
  title        = {Chiral random band matrices at zero energy},
  volume       = {114},
}

@article{Lacroix1984,
  author       = {Lacroix, Jean},
  date         = {1984},
  journaltitle = {Ann. Inst. H. Poincar\'{e} Sect. A},
  number       = {1},
  pages        = {97--116},
  title        = {Localisation pour l'op\'{e}rateur de {S}chr\"{o}dinger al\'{e}atoire dans un ruban},
  volume       = {40},
}

@article{KleinLacroixSpeis1990,
  author       = {Klein, Abel and Lacroix, Jean and Speis, Athanasios},
  date         = {1990},
  doi          = {10.1016/0022-1236(90)90031-F},
  journaltitle = {J. Funct. Anal.},
  number       = {1},
  pages        = {135--155},
  title        = {Localization for the {A}nderson model on a strip with singular potentials},
  volume       = {94},
}

@article{KusuokaStroock1985,
  author       = {Kusuoka, Shigeo and Stroock, Daniel},
  date         = {1985},
  doi          = {10.15083/00039520},
  journaltitle = {J. Fac. Sci. Univ. Tokyo Sect. IA Math.},
  number       = {1},
  pages        = {1--76},
  title        = {Applications of the Malliavin calculus. {II}},
  volume       = {32},
}

@article{Lessa2021,
  author       = {Lessa, Pablo},
  date         = {2021},
  doi          = {10.1090/btran/60},
  journaltitle = {Trans. Amer. Math. Soc. Ser. B},
  pages        = {105--129},
  title        = {Entropy and dimension of disintegrations of stationary measures},
  volume       = {8},
}

@incollection{Ledrappier-1982,
  author    = {Ledrappier, Fran\c{c}ois},
  editor    = {Hennequin, P. L.},
  location  = {Berlin},
  publisher = {Springer},
  booktitle = {\'{E}cole d'\'{E}t\'{e} de Probabilit\'{e}s de Saint-Flour XII---1982},
  date      = {1984},
  doi       = {10.1007/BFb0099434},
  pages     = {305--396},
  series    = {Lecture Notes in Mathematics},
  title     = {Quelques propri\'{e}t\'{e}s des exposants caract\'{e}ristiques},
  volume    = {1097},
}

@article{Margulis87,
  author       = {Gol'dsheid, I. Ya. and Margulis, G. A.},
  date         = {1987},
  journaltitle = {Dokl. Akad. Nauk SSSR},
  number       = {2},
  pages        = {297--301},
  title        = {The condition of simplicity for the spectrum of Lyapunov exponents},
  volume       = {293},
}

@article{Mirlin2000,
  author       = {Mirlin, Alexander D.},
  date         = {2000},
  journaltitle = {Phys. Rep.},
  number       = {5--6},
  pages        = {259--382},
  title        = {Statistics of energy levels and eigenfunctions in disordered systems},
  volume       = {326},
}

@article{Newman1986,
  author       = {Newman, Charles M.},
  date         = {1986},
  journaltitle = {Comm. Math. Phys.},
  number       = {1},
  pages        = {121--126},
  title        = {The distribution of {L}yapunov exponents: Exact results for random matrices},
  volume       = {103},
}

@article{SchlagShubinWolff2002,
  author       = {Schlag, Wilhelm and Shubin, Catherine and Wolff, Thomas},
  date         = {2002},
  journaltitle = {J. Anal. Math.},
  pages        = {173--220},
  title        = {Frequency concentration and localization lengths for the {A}nderson model at small disorders},
  volume       = {88},
}

@article{SchulzBaldes2004,
  author       = {Schulz-Baldes, Hermann},
  date         = {2004},
  journaltitle = {Geom. Funct. Anal.},
  number       = {5},
  pages        = {1089--1117},
  title        = {Perturbation theory for {L}yapunov exponents of an {A}nderson model on a strip},
  volume       = {14},
}

@article{ShubinWolff1998,
  author       = {Shubin, Catherine and Vakilian, Ramin and Wolff, Thomas},
  date         = {1998},
  journaltitle = {Geom. Funct. Anal.},
  number       = {5},
  pages        = {932--964},
  title        = {Some harmonic analysis questions suggested by {A}nderson--{B}ernoulli models},
  volume       = {8},
}

@article{GoldsheidMolchanovPastur1977,
  author       = {Goldsheild, Iosif Ya. and Molchanov, Stanislav A. and Pastur, Leonid A.},
  date         = {1977},
  journaltitle = {Funct. Anal. Appl.},
  number       = {1},
  pages        = {1--8},
  title        = {A pure point spectrum of the stochastic one-dimensional {S}chr\"{o}dinger operator},
  volume       = {11},
}

@article{PeledSchenkerShamisSodin2019,
  author       = {Peled, Ron and Schenker, Jeffrey H. and Shamis, Mira and Sodin, Sasha},
  date         = {2019},
  doi          = {10.1093/imrn/rnx145},
  journaltitle = {Int. Math. Res. Not. IMRN},
  number       = {4},
  pages        = {1030--1058},
  title        = {On the {W}egner orbital model},
  volume       = {2019},
}

@misc{DroginRBMLocalization2025,
  author      = {Drogin, Reuben},
  date        = {2025},
  doi         = {10.48550/arXiv.2508.05802},
  eprint      = {2508.05802},
  eprintclass = {math.PR},
  eprinttype  = {arXiv},
  title       = {Localization of one-dimensional random band matrices},
}

@article{Schenker,
  author       = {Schenker, Jeffrey H.},
  date         = {2009},
  doi          = {10.1007/s00220-009-0798-0},
  journaltitle = {Comm. Math. Phys.},
  pages        = {1065--1097},
  title        = {Eigenvector localization for random band matrices with power law band width},
  volume       = {290},
}

@misc{SchenkerPeledCipolloni,
  author      = {Cipolloni, Giorgio and Peled, Ron and Schenker, Jeffrey H. and Shapiro, Jacob},
  date        = {2022},
  doi         = {10.48550/arXiv.2206.05545},
  eprint      = {2206.05545},
  eprintclass = {math-ph},
  eprinttype  = {arXiv},
  title       = {Dynamical Localization for Random Band Matrices up to {$W\ll N^{1/4}$}},
}

@misc{SmartChen,
  author      = {Chen, Nixia and Smart, Charles K.},
  date        = {2022},
  doi         = {10.48550/arXiv.2206.06439},
  eprint      = {2206.06439},
  eprintclass = {math.PR},
  eprinttype  = {arXiv},
  title       = {Random band matrix localization by scalar fluctuations},
}

@article{BedrossianBlumenthalPunshonSmith2022,
  author       = {Bedrossian, Jacob and Blumenthal, Alex and Punshon-Smith, Samuel},
  date         = {2022},
  doi          = {10.1007/s00222-021-01069-7},
  journaltitle = {Invent. Math.},
  number       = {2},
  pages        = {429--516},
  title        = {A regularity method for lower bounds on the {L}yapunov exponent for stochastic differential equations},
  volume       = {227},
}

@misc{HaninJiang2025,
  author     = {Hanin, Boris and Jiang, Tianze},
  date       = {2025},
  eprint     = {2503.07872},
  eprinttype = {arXiv},
  title      = {Global universality of singular values in products of many large random matrices},
}

@misc{Goldstein2022RBM,
  author     = {Goldstein, Michael},
  date       = {2022},
  doi        = {10.48550/arXiv.2210.04346},
  eprint     = {2210.04346},
  eprinttype = {arXiv},
  title      = {Fluctuations and localization length for random band {GOE} matrix},
}

@misc{PillaiSmith2026Kac,
  author     = {Pillai, Natesh S. and Smith, Aaron},
  date       = {2026},
  eprint     = {2604.23828},
  eprinttype = {arXiv},
  title      = {{Kac}'s walk on rotation matrices mixes in {$n^2\log n$} steps},
}

@article{JitomirskayaSchulzBaldesStolz2003,
  author       = {Jitomirskaya, Svetlana and Schulz-Baldes, Hermann and Stolz, G\"{u}nter},
  date         = {2003},
  journaltitle = {Comm. Math. Phys.},
  number       = {1},
  pages        = {27--48},
  title        = {Delocalization in random polymer models},
  volume       = {233},
}

@article{HairerPunshonSmithRosatiYi,
  author       = {Hairer, Martin and Punshon-Smith, Sam and Rosati, Tommaso and Yi, Jaeyun},
  date         = {2026},
  journaltitle = {Duke Math. J.},
  note         = {To appear},
  title        = {Lower bounds on the top {L}yapunov exponent for linear {PDE}s driven by the 2D stochastic {N}avier--{S}tokes equations},
}

@unpublished{HairerStacy2026,
  author = {Hairer, Martin and Stacy, Declan},
  date   = {2026},
  note   = {Preprint},
  title  = {Asymptotics of {L}yapunov exponents and phase transitions for fluids with degenerate forcing},
}

@book{Ledoux2001,
  author    = {Ledoux, Michel},
  location  = {Providence, RI},
  publisher = {American Mathematical Society},
  date      = {2001},
  series    = {Mathematical Surveys and Monographs},
  title     = {The Concentration of Measure Phenomenon},
  volume    = {89},
}

\end{document}